\documentclass[11 pt]{amsart}

\usepackage[latin1]{inputenc}
\usepackage[T1]{fontenc}

\usepackage{amsfonts,amssymb,amsmath,amsthm}
\usepackage[headinclude,DIV13]{typearea}
\usepackage{hyperref}
\usepackage{geometry}
\usepackage{graphicx, psfrag}
\usepackage{array}

 \usepackage[dvipsnames]{color}
 \usepackage{setspace}

\newtheorem{lem}{Lemma}[section]
\newtheorem{prop}{Proposition}[section]

\newtheorem{theo}{Theorem}
\newtheorem{conj}{Conjecture}

\theoremstyle{definition} 
\theoremstyle{definition}

\renewcommand{\P}{\mathbb{P}}
\newcommand{\R}{\mathbb{R}}

\newcommand{\N}{\mathbb{N}}
\newcommand{\Z}{\mathbb{Z}}
\newcommand{\A}{\mathcal{A}}
\newcommand{\tb}[1]{\mathbf{#1}}
\newcommand{\wt}[1]{\widetilde{#1}} 

\newcommand{\eps}{\varepsilon}

\newcommand{\ben}{\vspace{0mm}\begin{equation}}
\newcommand{\een}{\vspace{0mm}\end{equation}}
\newcommand{\one}{\mathbf{1}}

\usepackage{ulem}

\numberwithin{equation}{section}

\begin{document}

\title[Speciation by mating preferences]{A stochastic model for speciation by mating preferences}

\author{Camille Coron}
\address{Laboratoire de Math\'ematiques d'Orsay, Univ. Paris-Sud, CNRS, Universit\'e Paris-Saclay, 91405 Orsay, France}
\email{camille.coron@math.u-psud.fr}

\author{Manon Costa}
\address{Institut de Math\'ematiques de Toulouse.  CNRS UMR 5219, Universit\'e
Paul Sabatier, 118 route de Narbonne, F-31062 Toulouse cedex 09}
\email{manon.costa@math.univ-toulouse.fr}

\author{H\'el\`ene Leman}
\address{CIMAT, De Jalisco S-N, Valenciana, 36240 Guanajuato, Gto., Mexico}
\email{helene.leman@polytechnique.edu}

\author{Charline Smadi}
\address{IRSTEA UR LISC, Laboratoire d'ing\'enierie des Syst\`emes Complexes, 9 avenue Blaise-Pascal CS
20085, 63178 Aubi\`ere, France and Complex Systems Institute of Paris île-de-France (ISC-PIF, UPS3611), 113 rue Nationale, Paris, France}
\email{charline.smadi@polytechnique.edu}

\begin{abstract}
Mechanisms leading to speciation are a major focus in evolutionary biology.
In this paper, we present and study a stochastic model of population where individuals, with type $a$ or $A$,
are equivalent from ecological, demographical and spatial points of view, and differ only by their mating 
preference: two individuals with the same genotype have a higher probability
to {mate and} produce a viable offspring. The population is subdivided in several patches and individuals may migrate between them.
We show that mating preferences by themselves, even if they are very small, are enough to entail reproductive isolation between patches,
and we provide the time needed for this isolation to occur {as a function of the population size}. Our results rely on a fine study of the stochastic
process and of its deterministic limit in large population, which is given by a system of coupled nonlinear differential equations.
Besides, we propose several generalisations of our model,
and prove that our findings are robust for those generalisations.

\end{abstract}

 \maketitle



\section*{Introduction}

Understanding mechanisms underlying speciation remains a central question in evolutionary biology. The main puzzle 
is the origin of isolating barriers that prevent gene flow among
populations {or within a population}.
Ecological speciation has been largely studied, highlighting the relations between sexual selection and speciation, and demonstrating negative links 
\cite{servedio2014counterintuitive,servedio2015effects} as well as beneficial ones \cite{boughman2001divergent}. 
Lande \cite{lande1981models} is the first one to have popularized the idea of sexual selection promoting speciation. 
Then numerous authors {have dealt} with it in depth 
{\cite{wu1985stochastic,turner1995model,higashi1999sympatric,van2004sympatric,Ritchie2007,pennings2008analytically}}. 
Furthermore, biological examples of speciation that involve well studied mechanisms of sexual selection are numerous and well documented, as the case of 
{Hawaiian} cricket \textit{Laupala} \cite{otte1989speciation,shaw2002divergence,mendelson2005sexual}, Amazonian frog \textit{Physalaemus} \cite{boul2007sexual}, or the cichlid fish species of Lake Victoria
\cite{seehausen2008speciation}.
Modelling approaches allow to investigate the relative roles of stochastic processes, ecological factors, and sexual selection in 
limiting gene flow. 
The role of so-called 'magic' or 'multiple effect' traits, which associate both adaptation to a new ecological 
niche and a mate preference as enhancer of speciation {has} been evidenced in many experimental studies \cite{merrill2012disruptive} as well as theoretical ones \cite{lande1988ecological,van1998sympatric}.
However, identifying the role of sexual selection itself as trigger of speciation without ecological adaptation has received 
less attention \cite{gavrilets2014review}, 
although some authors have illustrated the promoting role of sexual preference alone, using numerical simulations \cite{kondrashov1998origin,m2012sexual}.
In this paper, we aim at introducing and studying mathematically a stochastic model 
accounting for the stopping of gene flow between two subpopulations by means of sexual preference only.

We consider a population of {hermaphroditic} haploid individuals characterized by their genotype at one multi-allelic locus, 
and by their position on a space that is divided in several patches. This population is modeled by a multi-type birth and death 
process with competition, which is ecologically neutral in the sense that individuals with different genotypes are not 
characterized by different adaptations to environment or by different resource preferences. However, individuals reproduce sexually 
according to mating preferences that depend on their genotype: two individuals having 
the same genotype have a higher probability {to mate}. This assortative mating situation 
(assortative mating by phenotype matching)
has been highlighted notably in plant species, 
in particular due to simultaneous maturation of male and female reproductive organs \cite{Herrero2003, Savolainenetal2006}, 
and its selective advantages have been studied and modeled 
by Darwin \cite{Darwin1871} and more recently in the review \cite{JonesRatterman2009}. This review provides a detailed description of these models, 
as well as some empirical examples supporting mate preference evolution.
In addition to this sexual preference, individuals can migrate from one patch to another, at a rate depending on the frequency
of individuals carrying the other genotype and living in the same patch. 
Examples of animals migrating to find suitable mates are well documented \cite{Schwagmeyer1988, Honeretal2007}.
A migration mechanism similar to the one presented in our paper has been studied in \cite{PayneKrakauer1997} in a continuous space model.

The class of stochastic individual-based models with competition and varying population size we are studying
have been introduced in~\cite{bolker1997using,dieckmann2000relaxation} and made rigorous in a probabilistic setting in the
seminal paper of Fournier and M\'el\'eard~\cite{fournier2004microscopic}. Then they have been studied {by many authors (see}
\cite{champagnat2006microscopic,champagnat2006unifying,costa2015stochastic,leman2015convergence} {and references therein for instance)}. 
Initially restricted to asexual populations, such models have evolved to incorporate the case of sexual reproduction, in both haploid \cite{smadi2015eco} and 
diploid \cite{collet2011rigorous,coron2015slow,bovier2015survival} populations. 
{Taking into account varying population sizes and stochasticity is necessary if we aim at better understanding phenomena involving small populations{, like} mutational 
meltdown \cite{coron2013quantifying}, invasion of a mutant population \cite{champagnat2006microscopic}, {evolutionary suicide and rescue \cite{AbuAwadBilliard2017}} or {population} extinction time {(Theorem \ref{maintheo} of the current paper)}.}
In \cite{rudnicki2015model}, {Rudnicki and Zwole\'nski considered} both random and assortative mating in a phenotypically structured population. {In the present article, we consider a different kind of mechanism of sexual preference (see Section \ref{ModelDiscussion} for a detailed discussion), and our model is spatially structured.}

We study both the stochastic individual-based model and its 
deterministic limit in large population. 
We give a complete description of the equilibria of the limiting deterministic dynamical system, and prove that the stable equilibria are 
the ones where only one genotype 
survives in each patch. 
We use classical arguments based on Lyapunov functions \cite{lasalle1960some, chicone2006ode} to derive the convergence at exponential speed of the solution to one 
of the stable equilibria, depending on the initial condition. 
Our theoretical results hold for small migration rates but we conjecture using simulations that 
they hold for all the possible migration rates. 
This fine study of the large population limit is essential to derive the {average} behaviour of the stochastic process. 
{Then} using coupling techniques with branching processes, we derive bounds
for the time needed for speciation to occur in the stochastic process.
These bounds are explicit functions of the individual birth rate and the mating preference parameter.
Besides, we propose several generalisations of our model,
and prove that our findings are robust for those generalisations.

The structure of the paper is the following. In Section \ref{sectionmodel} we describe the model and present the main results.
Section \ref{ModelDiscussion} is devoted to a discussion on the biological assumptions of the model.
In Sections \ref{sectiondet} and \ref{sectionsto} we state properties of the deterministic limit and of the stochastic population process, respectively. They are key tools in the proofs of the main results, which are then completed.
In Section \ref{sectionillustr} we illustrate our findings and make conjecture on a more general result with the help of numerical simulations.
Section \ref{sectiongeneralisation} is devoted to some generalisations of the model. Finally, we state in the Appendix technical results needed in the proofs.

\section{Model and main results} \label{sectionmodel}
We consider a sexual haploid population with Mendelian reproduction (\cite{Griffithsetal2000}, chap. 3). 
{Time is continuous. At any moment, an individual can die, give birth or migrate.
As a consequence, generations are overlapping and there is no specific period for individuals to reproduce.}
Each individual carries an allele belonging to the genetic type space $\mathcal{A}:=\{A,a\}$, and lives in a 
patch $i$ in $\mathcal{I}=\{1,2\}$.
We denote by $\mathcal{E}=\mathcal{A}\times \mathcal{I}$ the type space, by $(\tb{e}_{\alpha,i}, (\alpha,i)\in \mathcal{E})$ 
the canonical basis of $\R^{\mathcal{E}}$, and by $\bar{\alpha}$ the complement of $\alpha$ in $\A$. 
The population is modeled by a multi-type birth and death process with values in $\mathbb{N}^{\mathcal{E}}$. 
More precisely, we denote by $n_{\alpha,i}$ the current number of $\alpha$-individuals in the patch $i$ and by $\tb{n}=(n_{\alpha,i}, (\alpha,i) \in\mathcal{E})$ the current state of the population.
{The birth rate is the consequence of the following mechanisms: 
at a rate $B>0$, any individual encounters another individual uniformly at random in its deme. 
Indeed, all the individuals are assumed to be ecologically and demographically equivalent, thus the probability that they are at the same place at the same time is uniform. 
Mathematically, the probability of encountering an individual of genotype $\alpha'$ in the patch $i$ writes
\ben
\frac{n_{\alpha',i}}{n_{\alpha,i}+n_{\bar{\alpha},i}},
\een
at the time of the encounter.
Then the probability that the encounter leads to a successful mating with the birth of an offspring is $b\beta/B\leq 1$ if the two individuals 
carry the same genotype, and $b/B \leq 1$ otherwise. As a consequence, the birth rate of individuals with genotype $\alpha$ in the deme 
$i$ is equal to 
\ben 
\label{birthrate}
\begin{aligned}
\lambda_{\alpha,i}(\tb{n})=&b\left(n_{\alpha,i}\beta\frac{ n_{\alpha,i}}{n_{\alpha,i}+n_{\bar{\alpha},i}}+\frac{1}{2}n_{\alpha,i}\frac{n_{\bar{\alpha},i}}{n_{\alpha,i}+n_{\bar{\alpha},i}}+\frac{1}{2}n_{\bar{\alpha},i}\frac{n_{\alpha,i}}{n_{\alpha,i}+n_{\bar{\alpha},i}}\right)\\&=
b n_{\alpha,i} \frac{\beta n_{\alpha,i}+n_{\bar{\alpha},i}}{n_{\alpha,i}+n_{\bar{\alpha},i}} .\end{aligned}\een}
{In other words, the parameter $\beta>1$ represents the 
"mating preference". Indeed, 
individuals meet uniformly at random and two encountering} individuals have a probability $\beta$ times larger to {mate and} 
give birth to a viable offspring if they carry the same 
allele $\alpha$. This modeling of mating preferences {, directly determined by the genome of each individual,} is biologically relevant, 
considering \cite{hollocher1997incipient} or \cite{haesler2005inheritance} for instance. \\

The death rate of $\alpha$-individuals in the patch $i$ writes
\begin{equation}\label{deathrate}
 d^K_{\alpha,i}(\tb{n})= \left(d+\frac{c}{K}(n_{\alpha,i}+n_{\bar{\alpha},i})\right) n_{\alpha,i},
\end{equation}
where $K$ is an integer accounting for the quantity of available resources or space. This parameter is related to the concept of carrying capacity, 
which is the maximum population size that the environment can sustain indefinitely, and is consequently a scaling parameter for the size of the community. The individual intrinsic death 
rate $d$ is assumed to be non negative and less than $b$:
\begin{equation}\label{hypothesebd}
            0\leq d<b.
          \end{equation}
The death rate definition \eqref{deathrate} implies that all the individuals are ecologically equivalent: the competition pressure does 
not depend on the alleles carried by the two individuals involved in an event of competition for food or space. The competition intensity is denoted by $c>0$.
Last, the migration of $\alpha$-individuals from patch $\bar{i}= \mathcal{I} \setminus  \{i\}$ to patch $i$ occurs at a rate
\begin{equation}\label{migrationrate}
 \rho_{\alpha,  \bar{i} \to i}(\tb{n})= p \left( 1-\frac{n_{\alpha,\bar{i}}}{n_{\alpha,\bar{i}}+n_{\bar{\alpha},\bar{i}}}\right)
 n_{\alpha,\bar{i}}=p\frac{n_{\alpha,\bar{i}}n_{\bar{\alpha},\bar{i}}}{n_{\alpha,\bar{i}}+n_{\bar{\alpha},\bar{i}}},
\end{equation}
(see Figure \ref{figspe}).
\begin{figure}[h]
    \centering
     \includegraphics[width=13cm,height=8cm]{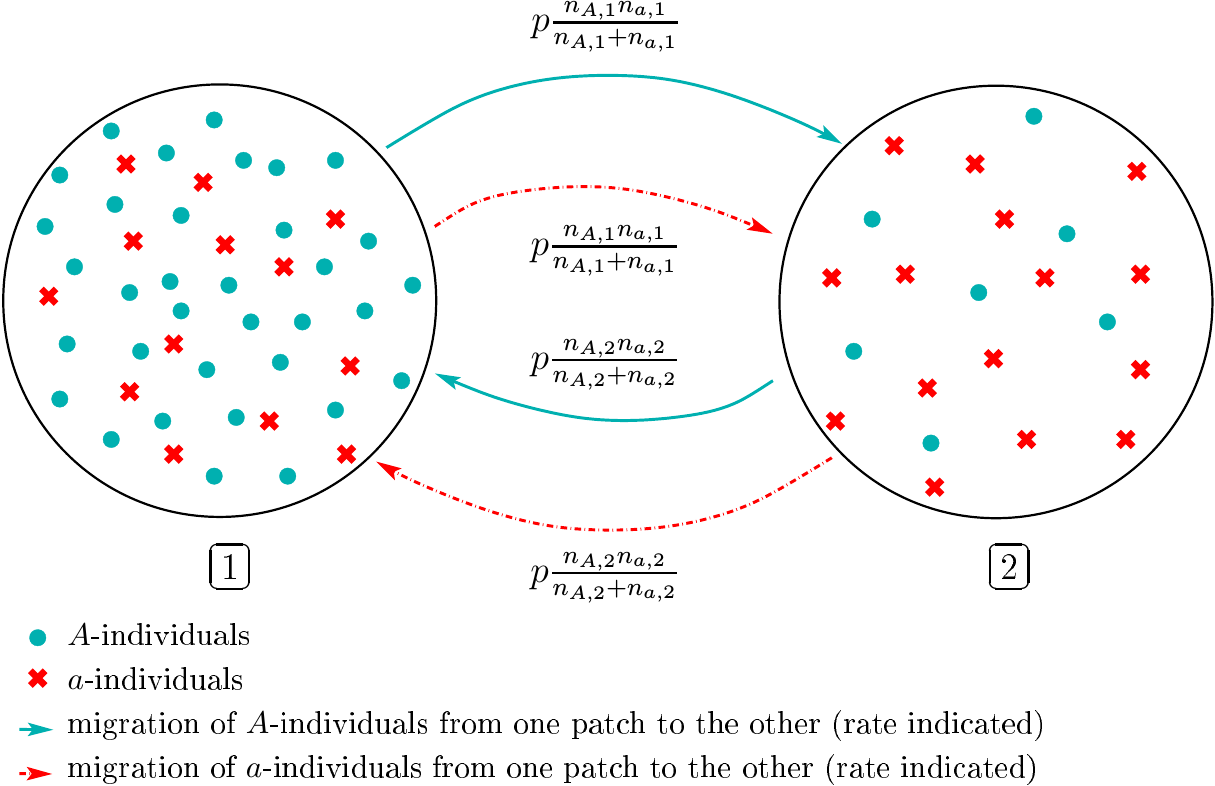}
     \caption{Migrations of $A$- and $a$-individuals between the patches.}
   \label{figspe}
\end{figure}
The individual migration rate of $\alpha$-individuals is proportional to the frequency of $\bar{\alpha}$-individuals in the patch. 
It reflects the fact that individuals prefer being in an environment with a majority of individuals of their own type. 
In particular, if all the individuals living in a patch 
are of the same type, there is no more migration outside this patch. 
Remark that the migration rate from patch $\bar{i}$ to $i$ is equal for {$A$}- and {$a$}-individuals, 
hence to simplify notation, we denote 
$$\rho_{\bar i\to  i}(\tb{n})=\rho_{A,\bar i\to i}(\tb{n})=\rho_{a,\bar i\to i}(\tb{n}).$$ 
A biological discussion of the model is provided in Section \ref{ModelDiscussion}. Besides, extensions of this model are presented and studied in Section \ref{sectiongeneralisation}.

The community is therefore represented at {every} time $t\geq0$ by a stochastic process with values in $\R^\mathcal{E}$: 
$$(\tb{N}^K(t),t\geq0)=(N^K_{\alpha,i}(t), (\alpha,i)\in\mathcal{E}{,t\geq0)},$$
whose transitions are, for $\tb{n} \in \N^\mathcal{E}$ and $(\alpha,i)\in\mathcal{E}$:
\begin{center}
\begin{tabular}{lllll}
$\tb{n}$
& $\longrightarrow $& $\tb{n}+\tb{e}_{\alpha,i}$ & at rate&$\lambda_{\alpha,i}(\tb{n})$,\\ 
& $\longrightarrow$ & $\tb{n}-\tb{e}_{\alpha,i}$ & at rate&$d^K_{\alpha,i}(\tb{n})$,\\ 
& $\longrightarrow$ & $ \tb{n} + \tb{e}_{\alpha,i} -\tb{e}_{\alpha,\bar{i}}$ & at rate&
 $ \rho_{\bar i\to i}(\tb{n}) $.
\end{tabular}\end{center}

As originally done by Fournier and M\'el\'eard \cite{fournier2004microscopic}, it is convenient to
represent a trajectory of the process $\tb{N}^K$ as the unique solution of a system of stochastic differential equations 
driven by Poisson point measures. We introduce twelve independent Poisson point measures $(R_{\alpha,i}, M_{\alpha,i}, D_{\alpha,i},(\alpha,i)\in\mathcal{E})$ on 
$\R_+^2$ with intensity $ds\hspace{0.1cm}d\theta$. These measures represent respectively the birth, migration and death events in the population $N^K_{\alpha,i}$.
We obtain for every $t\geq0$,
\begin{equation}
\label{def_poisson}
\begin{aligned}
\tb{N}^K(t)=\tb{N}^K(0)+\sum_{(\alpha,i)\in\mathcal{E}}\Big[&\int_0^t \int_0^\infty \tb{e}_{\alpha,i}
\one_{\{\theta\leq \lambda_{\alpha ,i}(\tb{N}^K(s-))\}}R_{\alpha,i}(ds,d\theta)\\
&-\int_0^t \int_0^\infty  \tb{e}_{\alpha,i}\one_{\{\theta\leq d^K_{\alpha,i}(\tb{N}^K(s-))\}}D_{\alpha,i}(ds,d\theta)\\
&+\int_0^t \int_0^\infty  (\tb{e}_{\alpha,\bar i}-\tb{e}_{\alpha,i})\mathbf{1}_{\{\theta\leq \rho_{\bar i\to i}(\tb{N}^K(s-))\}}M_{\alpha,i}(ds,d\theta)
\Big].
\end{aligned}
\end{equation}

In the sequel, we will assume that the initial population sizes $(N^K_{\alpha,i}(0),(\alpha,i)\in\mathcal{E})$ are of order $K$.
 As a consequence, we consider a rescaled stochastic process 
$$(\tb{Z}^K(t),t\geq0)=(Z^K_{\alpha,i}(t),(\alpha,i)\in\mathcal{E},t\geq0)=\left(\frac{\tb{N}^K(t)}{K},t\geq0\right),$$
which will be comparable to a solution of the dynamical system
\ben\label{systdet}
\left\{\begin{array}{l} \frac{d}{dt}z_{A,1}(t)=
z_{A,1}\Bigl[ b\frac{\beta z_{A,1}+z_{a,1}}{z_{A,1}+z_{a,1}}-d-c(z_{A,1}+z_{a,1})-p\frac{z_{a,1}}{z_{A,1}+z_{a,1}}\Bigr]+p\frac{z_{A,2}z_{a,2}}{z_{A,2}+z_{a,2}}\\
\frac{d}{dt}z_{a,1}(t)=
z_{a,1}\Bigl[ b\frac{\beta z_{a,1}+z_{A,1}}{z_{A,1}+z_{a,1}}-d-c(z_{A,1}+z_{a,1})-p\frac{z_{A,1}}{z_{A,1}+z_{a,1}}\Bigr]+p\frac{z_{A,2}z_{a,2}}{z_{A,2}+z_{a,2}}\\
\frac{d}{dt}z_{A,2}(t)=
z_{A,2}\Bigl[ b\frac{\beta z_{A,2}+z_{a,2}}{z_{A,2}+z_{a,2}}-d-c(z_{A,2}+z_{a,2})-p\frac{z_{a,2}}{z_{A,2}+z_{a,2}}\Bigr]+p\frac{z_{A,1}z_{a,1}}{z_{A,1}+z_{a,1}}\\
\frac{d}{dt}z_{a,2}(t)=
z_{a,2}\Bigl[ b\frac{\beta z_{a,2}+z_{A,2}}{z_{A,2}+z_{a,2}}-d-c(z_{A,2}+z_{a,2})-p\frac{z_{A,2}}{z_{A,2}+z_{a,2}}\Bigr]+p\frac{z_{A,1}z_{a,1}}{z_{A,1}+z_{a,1}}.\end{array}\right.
\een
{Note that, from a mathematical point of view, it is possible to reduce the number of parameters $b$, $c$, $d$,
$p$, $\beta$. Using a time scaling and a size scaling, we can prove that only
three effective parameters are necessary to describe the mathematical
behaviour of the system, corresponding to a reformulation of the parameters $\beta$, $d$ and $p$ (we refer the interested reader to the appendix for more details). 
However, since each parameter has a biological meaning, we will keep these notations.} 

Let us denote by 
$$(\tb{{z}}^{(\bf{z}^0)}(t),t\geq0)=({z}^{(\bf{z}^0)}_{\alpha,i}(t),(\alpha,i)\in \mathcal{E})_{t\geq0}$$ 
the unique solution to \eqref{systdet} starting from 
$\bf{z}(0)=\bf{z}^0 \in \R_+^\mathcal{E}$. The uniqueness derives from the fact that the vector field is locally lipschitz and that the solutions do not 
explode in finite time \cite{chicone2006ode}. 
We have the following classical approximation result which will be proven in 
Appendix \ref{appendix}:

\begin{lem}\label{lemapprox}
 Let $T$ be in $\R_+^*$. Assume that the sequence $(\tb{Z}^K(0),K \geq 1)$ converges in probability when $K$ goes to infinity to a deterministic vector ${\bf{z}^0} \in \R_+^\mathcal{E}$.
 Then 
\begin{equation}\label{EK2}
\underset{K \to \infty}{\lim}\  \sup_{s\leq T}\ \|\tb{Z}^K(s)-\tb{z}^{(\bf{z}^0)}(s)  \|=0 \quad \text{in probability},
\end{equation}
where $\|. \|$ denotes the $L^\infty$-Norm on $\R^\mathcal{E}$.
\end{lem}

When $K$ is large, this convergence result allows one to derive the global behaviour of the population process $\tb{N}^K$ 
from the behaviour of the {dynamical} system \eqref{systdet}. 
Therefore, a fine study of \eqref{systdet} is needed. To this aim, let us introduce the parameter 
\begin{equation}\label{defz}
 \zeta:= \frac{\beta b-d}{c},
\end{equation}
which corresponds to the equilibrium {size} of the $\alpha$-population for the dynamical system \eqref{systdet}, in a patch with no $\bar{\alpha}$-individuals 
and no migration. Let us also define the parameters 
\begin{equation}\label{defzbar}
 \wt{\zeta}:=\dfrac{b^2(\beta^2-1)+2p(b-d)-2bd(\beta-1)}{4c(b(\beta-1)+p)} \quad \text{and}
\quad 
 \Delta := \zeta\left(\zeta-2p\frac{\wt{\zeta}}{b(\beta-1)+p}\right) > 0
\end{equation}
(see \eqref{Deltapositif} for the positivity of $\Delta$). We derive in Section \ref{sectiondet} the following properties of the dynamical system \eqref{systdet}:

\begin{theo}
\label{thm_systdet}
\begin{enumerate}
\item
For $\beta \geq 1$, the following points for which only one type remains, in only one patch
\begin{equation}
 \label{eq-z000}
(\zeta,0,0,0) \quad (0,\zeta,0,0)\quad (0,0,\zeta,0) \quad (0,0,0,\zeta)
\end{equation}
are non-null and non-negative equilibria of the dynamical system \eqref{systdet}.
\item For $\beta > 1$, the remaining non-null and non-negative fixed points are exactly:
\begin{itemize}
\item Equilibria for which each type is present in exactly one patch 
\begin{equation}
 \label{eq-z00z}
(\zeta,0,0,\zeta),\quad (0,\zeta,\zeta,0)
\end{equation}
\item Equilibria for which only one type remains present, in both patches
\begin{equation}
 \label{eq-z0z0}
(\zeta,0,\zeta,0), \quad (0,\zeta,0,\zeta)
\end{equation}
\item Equilibria with both types remaining in both patches
\begin{equation}
 \label{eq-zzzz}
\left(\frac{b(\beta +1)-2d}{4c},\frac{b(\beta +1)-2d}{4c},\frac{b(\beta +1)-2d}{4c},\frac{b(\beta +1)-2d}{4c}\right)
\end{equation}
\begin{align}
\label{eq-complique}
&\Bigl(\frac{\zeta+\sqrt{\Delta}}{2},\frac{\zeta-\sqrt{\Delta}}{2},\wt{\zeta},\wt{\zeta} \Bigr), \quad \Bigl(\frac{\zeta-\sqrt{\Delta}}{2},\frac{\zeta+\sqrt{\Delta}}{2},\wt{\zeta},\wt{\zeta} \Bigr),
\\ &\Bigl(\wt{\zeta},\wt{\zeta},\frac{\zeta+\sqrt{\Delta}}{2},\frac{\zeta-\sqrt{\Delta}}{2} \Bigr), \quad \Bigl(\wt{\zeta},\wt{\zeta},\frac{\zeta-\sqrt{\Delta}}{2},\frac{\zeta+\sqrt{\Delta}}{2} \Bigr).
\end{align}
\end{itemize}
The only stable equilibria of the dynamical system \eqref{systdet} are those defined in Equation \eqref{eq-z00z}, 
for which each of the two alleles is present in exactly one patch, and those given in Equation \eqref{eq-z0z0} for which only one type remains.\\
\item For $\beta = 1$, the remaining non-null and non-negative fixed points are exactly the two sets 
$$\mathcal{L}=\{\mathbf{u}(x)=(\zeta-x,x,x,\zeta-x),x\in [0,\zeta]\}$$ 
and 
$$\tilde{\mathcal{L}}=\{\tilde{\mathbf{u}}(x)=(\zeta-x,x,\zeta-x,x),x\in [0,\zeta]\}.$$ 
Those equilibria are non-hyperbolic. For any $x\in[0,\zeta]\setminus \{\zeta/2\}$, the Jacobian matrix at 
the equilibrium $\mathbf{u}(x)$ admits $0$ as an eigenvalue (associated with the eigenvector  
$(1,-1,-1,1)$, direction of the line $\mathcal{L}$) and three negative eigenvalues. Some symmetrical results hold for $\tilde{\mathbf{u}}(x)$. 
The Jacobian matrix at the equilibrium $\mathbf{u}(\zeta/2)=\tilde{\mathbf{u}}(\zeta/2)$ admits two negative eigenvalues and the eigenvalue $0$ which 
is of multiplicity two.
\end{enumerate}
\end{theo}

The equilibria \eqref{eq-z00z} and \eqref{eq-z0z0} correspond to the case where reproductive isolation occurs since the gene flow between the 
two patches ends to be null.
The dynamics of the solutions are {fundamentally} different in the cases $\beta>1$ and $\beta=1$. 
They converge to an equilibrium without gene flow when $\beta>1$, whereas
when $\beta=1$, depending on the initial condition, the solutions will converge to different equilibria with a nonzero migration rate, 
that is without reproductive isolation.
The following proposition states that for each $x$, we can construct particular trajectories of the system which converge to $\mathbf{u}(x)$.
\begin{prop}
\label{prop1}
Let us introduce for any $w\in(0,+\infty)$ and $x\in[0,w]$ the vector
$$\mathbf{v}(w,x)=(w-x,x,x,w-x).$$
The solution $z^{(\mathbf{v}(w,x))}$ of the system \eqref{systdet} with $\beta=1$ such that $z^{(\mathbf{v}(w,x))}(0)=\mathbf{v}(w,x)$ converges when $t\to\infty$ to the equilibrium 
$\mathbf{u}(\zeta x/w)$.
\end{prop}
In particular, the equilibria \eqref{eq-z00z} are not asymptotically stable when $\beta=1$ since solutions starting in any neighbourhood of 
\eqref{eq-z00z} can converge to different equilibria. Note that the shape of the migration is not sufficient to entail reproductive isolation although it 
seems to reinforce the homogamy described by the $\beta$ parameter. 
Thanks to simulations in Section~\ref{sectionillustr}, we will 
see that the effect of migration on the system dynamics is rather involved.

As a consequence, we assume $\beta>1$ in the sequel. The following theorem gives the long-time convergence of the dynamical system \eqref{systdet} 
toward a stable equilibrium of interest, when starting from an explicit subset of $\R_+^\mathcal{E}$.
To state this latter, we need to define
the subset of $\R_+^\mathcal{E}$
\begin{equation} \label{defdelta} \mathcal{D}:=\{ \tb{z} \in \R_+^\mathcal{E}, z_{A,1}-z_{a,1}>0, z_{a,2}-z_{A,2}>0 \},
\end{equation}
and the positive real number 
\begin{equation}\label{defp0}
 p_0=\frac{\sqrt{b(\beta-1)[b(3\beta+1)-4d]}-b(\beta-1)}{2}.
\end{equation}
Notice that under Assumption \eqref{hypothesebd} and as $\beta>1$,
 $$ p_0 < b(\beta+1)-2d. $$
Finally, for $p <b(\beta+1)-2d$, we introduce the set
\begin{equation}\label{defK}
 \mathcal{K}_p:= \left\{ 
\tb{z} \in \mathcal{D}, \; \{ z_{A,1}+z_{a,1}, \ z_{A,2}+z_{a,2} \} \in \left[ \frac{b(\beta+1)-2d-p}{2c},  
\frac{2b\beta-2d+p}{2c}\right] 
\right\}.
\end{equation}
Then we have the following result:

\begin{theo}\label{theoCvceD}
Let $p<p_0$. Then 
\begin{itemize}
 \item Any solution to~\eqref{systdet} which starts from
$ \mathcal{D}$
converges to the equilibrium $(\zeta,0,0,\zeta)$. 
\item If the initial condition of \eqref{systdet} lies in $\mathcal{K}_p$, there exist two positive 
constants $k_1$ and $k_2$, depending on the 
 initial condition, such that for every $t \geq 0$,
 $$ \|\tb{z}(t)- (\zeta,0,0,\zeta)\|\leq k_1 e^{- k_2 t}. $$
\end{itemize}
Symmetrical results hold for the equilibria $(0,\zeta,\zeta,0)$, $(\zeta,0,\zeta,0)$ and $(0,\zeta,0,\zeta)$.
\end{theo}

Note that the limit reached 
 depends on the genotype which is initially in majority in each patch,
since the subset $\mathcal{D}$ is invariant under the dynamical system~\eqref{systdet}. 
Secondly, when $p=0$, the results of Theorem~\ref{theoCvceD} can be proven easily since the two patches are independent from each other. 
The difficulty is thus to prove the result when $p>0$. Our argument allows us to deduce an explicit constant $p_0$ 
under which we have convergence to an equilibrium with reproductive isolation between patches. 
However, we are not able to deduce a rigorous result for all $p$. Indeed, when $p$ increases, 
there are more mixing between the two patches which makes the model difficult to study.
Nevertheless simulations in Section \ref{sectionillustr} suggest that the result stays true.\\

Let us now introduce our main result on the probability and the time needed for the stochastic process $\mathbf{N}^K$ to reach 
a neighbourhood of the equilibria defined in \eqref{eq-z00z}.

\begin{theo}\label{maintheo}
 Assume that $\tb{Z}^K(0)$ converges in probability to a deterministic vector ${\bf{z}^0}$ belonging to 
$\mathcal{D}$, with $(z_{a,1}^0,z_{A,2}^0)\neq (0,0)$.
Introduce the following bounded set depending on $\eps>0$:
 $$ \mathcal{B}_\eps:= [(\zeta-\eps)K,(\zeta+\eps)K] \times \{0\} \times \{0\} \times [(\zeta-\eps)K,(\zeta+\eps)K]. $$
 Then there exist three positive constants $\varepsilon_0$, $C_0$ and $m$, and a positive constant 
 $V$ depending on $(m,\eps_0)$ such that if $p < p_0$ and $\eps\leq \eps_0$,
 \begin{equation}
\label{eq_maintheo}
  \lim_{K \to \infty}\P \left( \left| \frac{T^K_{\mathcal{B}_{\eps}}}{\log K}-\frac{1}{b(\beta-1)} \right|\leq C_0\eps, 
 \; \tb{N}^K\left(T^K_{\mathcal{B}_{\eps}}+t\right)\in  \mathcal{B}_{m\eps}\; \forall t \leq e^{VK} \right)= 1,
 \end{equation}
 where $T^K_\mathcal{B}$, $\mathcal{B} \subset \R_+^\mathcal{E}$ is the hitting time of the set $\mathcal{B}$ by the population process 
 $\tb{N}^K$.\\
 Symmetrical results hold for the equilibria $(0,\zeta,\zeta,0)$, $(\zeta,0,\zeta,0)$ and $(0,\zeta,0,\zeta)$.
 \end{theo}
 
{This theorem gives the order of magnitude of the time to reproductive isolation between the two patches, as a function of the population size scaling factor $K$. This isolation time is infinite when considering the dynamical system \eqref{systdet} for which $K$ is equal to infinity.} 
Note that the time needed to reach the reproductive isolation is inversely proportional to $\beta-1$ which, as studied previously, suggests that 
the system behaves differently for $\beta=1$. Moreover, the time does not depend on the parameter $p$. 
{Intuitively, this can be understood as follows:} the time needed to reach a neighbourhood of the state $(\zeta,0,0,\zeta)$ is of 
order $1$, and from this neighbourhood the time needed for the 
complete extinction of the $a$-individuals in the patch $1$ and the $A$-individuals in the patch $2$ is much longer, it is of order $\log K$. 
During this second phase, the migrations between the two patches are already balanced, {which entails the independence with respect to $p$}. 
Furthermore, the constant does not depend on $d$ and $c$ since there is no ecological difference between the two types and the two patches: 
during the second phase, the natural 
birth rate of the $a$-individuals in the patch $1$ is approximately $b$ since 
the patch $1$ is almost entirely filled with $A$-individuals,
and their natural 
death rate can be approximated by $d+c\zeta=b\beta$ where the term $c\zeta$ comes from the competition exerted by the $A$-individuals. Thus, their natural growth 
rate is approximately $b-b\beta$ which only depends on the birth parameters.

Note that Theorem~\ref{maintheo} gives not only an estimation of the time to reach a neighbourhood of the limit, but also it proves that 
the dynamics of the population process stays a long time in the neighbourhood of equilibria \eqref{eq-z00z} after this time. 

Finally, the assumption $(z_{a,1}^0,z_{A,2}^0)\neq (0,0)$ 
is necessary to get the lower bound in~\eqref{eq_maintheo}. Indeed, 
if $(z_{a,1}^0,z_{A,2}^0)= (0,0)$, the set $\mathcal{B}_\eps$ is reached faster, and thus only the upper bound still holds. 
In this case, the speed to reach the set $\mathcal{B}_\eps$ will depend on the speed of convergence of the sequence $(Z^K_{a,1},Z^K_{A,2})$ to the 
limit $(0,0)$. In the trivial example where $(Z^K_{a,1},Z^K_{A,2})=(0,0)$, $T^K_{\mathcal{B}}$ will be of order $1$ which is the time needed for the processes 
$Z^K_{A,1}$ and $Z^K_{a,2}$ to 
reach a neighbourhood of the equilibrium $\zeta$.

\section{Discussion of the model}\label{ModelDiscussion}

{Assortative mating and genetic incompatibilities have been modeled and studied by many authors
 using discrete time models (see for instance \cite{gavrilets1998evolution, matessi2002long, gavrilets2003perspective, BurgerSchneider2006, Servedio2010} and references therein). 
 Comparing continuous time models and discrete non-overlapping generations models is tricky. Indeed, some concepts that are clearly defined for the second 
 class of models, like mating success or cost of choosiness, are hard to adapt to the first one. In this section we discuss our model in link 
 with previous work.}

\subsection*{Assortative mating}
Assortative mating can result from different factors. {Here, we are interested in assortative mating by phenotypic matching. 
That is to say, we consider uniform encounters between individuals and assume that assortative mating is the consequence of
an increased mating probability between individuals with the same phenotype, when encountering. This leads to the following birth rate of $\alpha$-individuals on patch  $i$
\ben \label{Assortmate1}bn_{\alpha,i}\frac{\beta n_{\alpha,i}+n_{\bar{\alpha},i}}{ n_{\alpha,i}+n_{\bar{\alpha},i}}.\een
We think of a comportemental or a mechanical prezygotic isolation after encountering. As an example of our birth rate definition, we can think of high density 
populations of milkweed longhorn beetle \textit{Tetraopes tetraophthalmus} where assortative mating is strong because at high density, large males are more 
likely to interfere with small males' copulation with large females \cite{mclain1987male}. 
We can find other examples of this type in a recent review on assortative mating in animals \cite{jiang2013assortative}.
Note that we can also interpret the birth rate \eqref{Assortmate1} as
post-zygotic isolation \cite{ravigne2010speciation}, thinking of a
low survival probability of the diploid zygotes of genotype $Aa$ after mating \cite{gavrilets2004fitness,bank2012limits}.}

{In contrast with our model, most papers about sexual preferences 
(see for instance \cite{gavrilets1998evolution, matessi2002long, BurgerSchneider2006, Servedio2010}) use generational models 
with infinite population size, and study the evolution through time of the frequency of each genotype.
As a consequence, they express the population dynamics in terms of a table describing the frequencies of mating at each generation.
With our notations, the table of the Supplementary Material of \cite{Servedio2010} giving probabilities 
that the individuals with genotype $\alpha$ mate with any individual with genotype 
$\alpha'$ in the deme $i$ and transmit their genotype writes:
\begin{equation}\label{Probamate}
\begin{tabular}{c|c|c}
 $\alpha$ $\smallsetminus$ $\alpha'$ &$A$  &  $a$    \\\hline
$A$    &  $\frac{ \beta n^2_{A,i}}{(n_{A,i}+n_{a,i})(\beta n_{A,i}+n_{a,i})}$ &  $\frac{  n_{A,i} n_{a,i} }{(n_{A,i}+n_{a,i})(\beta n_{A,i}+n_{a,i})}$ \\ \hline
$a$        &  $\frac{ n_{A,i} n_{a,i}}{(n_{A,i}+n_{a,i})( n_{A,i}+\beta n_{a,i})}$ &  $\frac{ \beta n^2_{a,i}}{(n_{A,i}+n_{a,i})(n_{A,i}+\beta n_{a,i})}$  \\ 
\end{tabular}
\end{equation}
Here the lines give the genotype $\alpha$ transmitted to the offspring (often called the female genotype).
Note that the same probabilities are derived in \cite{gavrilets1998evolution} (in the case $n=\infty$), or in \cite{matessi2002long}.
These mating probabilities at first glance may seem very different from the equations governing the births in our model \eqref{Assortmate1}.
However, as explicited in the Supplementary Information of \cite{Servedio2010}, the mate choice mechanism is similar to ours: 
during mating, individuals encounter uniformly and are more likely to mate with an individual with the allele that they themselves carry.}

{The difference comes from the fact that
the models in 
\cite{gavrilets1998evolution, matessi2002long, BurgerSchneider2006, Servedio2010} are in discrete time whereas ours is in continuous time.
Moreover, they assume that all the females reproduce once. In our model, we define the rates of mating, birth and death and not the probabilities of these events. 
That is to say, any individual can reproduce many times or even die before having any chance to reproduce. 
To compare our model with a generational one, we can compute the probabilities for a given $\alpha$-individual in the deme $i$ to reproduce
with an $\alpha'$-individual and transmits its genotype at time $t$,  
conditionally to the fact that this 
$\alpha$-individual reproduces and transmits its genotype at time $t$.
We get:
\begin{equation*}
\begin{aligned}
\P(  \alpha   \text{ mate with any } \alpha \text{ and transmits at } & t  \lvert \alpha \text{ reproduces  and transmits at } t )\\
&= \frac{\P( \alpha  \text{ mate with any } \alpha \text{ and transmits at }  t ) }{\P(\alpha \text{ reproduces and transmits at } t)}\\
&= \frac{ \frac{b\beta}{B}  \frac{n_{\alpha,i}}{(n_{\alpha,i}+n_{\bar{\alpha},i}) }  }{ \frac{b\beta}{B}  \frac{ n_{\alpha,i}}{(n_{\alpha,i}+n_{\bar{\alpha},i})} +\frac{b}{B} \frac{ n_{\bar{\alpha},i}}{(n_{\alpha,i}+n_{\bar{\alpha},i})}  }
= \frac{\beta n_{\alpha,i}}{\beta n_{\alpha,i}+n_{\bar{\alpha},i}},
\end{aligned}
\end{equation*}
and
\begin{equation*}
\P(   \alpha   \text{ mate with any } \bar\alpha \text{ and transmits at } t  \lvert \alpha \text{ reproduces at } t \text{  and transmits})
= \frac{ n_{\bar{\alpha},i} }{\beta n_{\alpha,i}+n_{\bar{\alpha},i}}.
\end{equation*}
Multiplying by the frequency of the $\alpha$-individuals in the deme $i$ gives the expressions 
derived in classical generational models \eqref{Probamate}. 
That is to say, in both cases, the mating probabilities at the mating time are similar.}

{Assortative mating can also derive from a non-homogeneous mating of individuals as proposed by \cite{rudnicki2015model}. In their case, 
the birth rate for an individual with genotype $\alpha$ in the deme $i$ is
\ben \label{Assortmate2}b n_{\alpha,i} \frac{\beta n_{\alpha,i} + \frac{1}{2} n_{\bar{\alpha},i}}{\beta n_{\alpha,i} + n_{\bar{\alpha},i}} +
b n_{\bar{\alpha},i} \frac{\frac{1}{2} n_{\alpha,i}}{n_{\alpha,i} + \beta n_{\bar{\alpha},i}}. \een 
This non uniform encountering and mating can presume of an ecological or temporal isolation of reproducing individuals. Indeed, 
with this expression, individuals of the same genotype are more likely to encounter than individuals of different genotypes, as if individuals of the 
same genotype were more likely to be at the same place at the same time.
As an example, this definition of birth rate can model reproduction of herma\-phro\-di\-tic plants with uniform pollen dispersal within each deme and 
simultaneous maturation of both male and female reproductive organs, at a time that depends on the plant genotype, as studied 
in \cite{Herrero2003, Savolainenetal2006}.}

\subsection*{Cost of choosiness}
Cost of choosiness for populations having 
specific mating periods and limited mating trials have been studied in \cite{gavrilets1998evolution, BurgerSchneider2006, KoppHermisson2008} notably. {In these articles, each female can reproduce at most once, and cost of choosiness is quantified by a maximum number $n$ of encounters that a female can make, in order to reproduce. In the present article, we assume a constant 
availability of both male and female organs of hermaphroditic individuals (like in sponges, sea anemones, tapeworms, snails, earthworms, or 
some fishes \cite{avise2009evolutionary} for instance). However, the potential of reproduction of each individual is hampered by its lifespan, which is stochastic.} 
\subsection*{Initial conditions}
Concerning the initial allelic diversity, which is a question highly debated in the literature on speciation \cite{weissing2011adaptive}, 
we have in mind populations where traits evolved neutrally before taking part in mating preferences after 
a change in the environment or a migration of the population to a new environment.
For example, it is the case for the two sister species \textit{P. nyererei} and \textit{P. pundamilia}. Males of these two species have different nuptial colorations 
(red and blue, respectively), and females of these two species have preferences for a specific male nuptial coloration in clear water 
(red for \textit{P. nyererei} and blue for \textit{P. pundamilia}).
These mating preferences have been proven to be inheritable \cite{haesler2005inheritance}, and uniformly random mating in turbid water has been inferred 
from phenotype frequency distribution in 
nature \cite{seehausen1997cichlid}.

\subsection*{Migration}
In our model the migration rate of a given individual is proportional to the {frequency} of individuals that do not have the same genotype as the considered 
individual. 
The idea is that an individual is more prone to move if it does not find suitable mates in its deme.
This particular form of mating success dependent dispersal has also been studied in \cite{PayneKrakauer1997} for a continuous space. 
{In \cite{chaput2010condition}, the authors study the dispersal behaviour of the banded damselflies \textit{Calopteryx spendens}, which display a lek mating 
system. They observe that females move to find a suitable mate, and that they disperse less when the sex-ratio is male biased.
This is in agreement with our hypothesis that individuals migration rate is a decreasing function of the frequency of suitable mates. 
More generally,}
correlations between male dispersal and mating success have been empirically observed {(see \cite{Schwagmeyer1988} or \cite{Honeretal2007} for instance)}.

{To emphasize the fact that migration is also governed by mate choice, we could rewrite the parameter of migration $p$ as $(\beta-1)p'$. Some formulations 
would then be modified although results would be identical. 
A degree of freedom would be kept thanks to $p'$ showing that migration and mating choice are not completely linked in our model,
but we could not have systems with migration and no preference ($\beta=1$) anymore.
In Section \ref{sectionillustr}, we provide a deep study of the influence of the parameter $p$ on the behaviour of the system.}

\section{Study of the dynamical system} \label{sectiondet}

In this section, we study the dynamical system~\eqref{systdet} in order to prove Theorems~\ref{thm_systdet} and \ref{theoCvceD}. 
In the first subsection, we are concerned with the equilibria of~\eqref{systdet} 
and their local stability (Theorem \ref{thm_systdet}). In the second subsection, we look more closely at the case where the migration 
rate $p$ is lower than $p_0$ and prove the convergence of the solution 
to~\eqref{systdet} towards one of the equilibria with an exponential rate once the trajectory belongs 
to $\mathcal{K}_p$ (Theorem \ref{theoCvceD}).

\subsection{Fixed points and stability when $\beta>1$} \label{sectionfixedpoints}

First of all, we prove that all nonnegative and non-zero stationary points of~\eqref{systdet} are given in Theorem~\ref{thm_systdet}.
Let us write the four equations defining equilibria $(z_{A,1},z_{a,1},z_{A,2},z_{a,2})$ of the dynamical system \eqref{systdet}:
\begin{align}
\label{equa1}
&z_{A,1}\Bigl[ b\frac{\beta z_{A,1}+z_{a,1}}{z_{A,1}+z_{a,1}}-d-c(z_{A,1}+z_{a,1})-p\frac{z_{a,1}}{z_{A,1}+z_{a,1}}\Bigr]+p\frac{z_{A,2}z_{a,2}}{z_{A,2}+z_{a,2}}=0, \\
\label{equa2}
& z_{a,1}\Bigl[ b\frac{\beta z_{a,1}+z_{A,1}}{z_{A,1}+z_{a,1}}-d-c(z_{A,1}+z_{a,1})-p\frac{z_{A,1}}{z_{A,1}+z_{a,1}}\Bigr]+p\frac{z_{A,2}z_{a,2}}{z_{A,2}+z_{a,2}}=0,\\
\label{equa3}
& z_{A,2}\Bigl[ b\frac{\beta z_{A,2}+z_{a,2}}{z_{A,2}+z_{a,2}}-d-c(z_{A,2}+z_{a,2})-p\frac{z_{a,2}}{z_{A,2}+z_{a,2}}\Bigr]+p\frac{z_{A,1}z_{a,1}}{z_{A,1}+z_{a,1}}=0,\\
\label{equa4}
& z_{a,2}\Bigl[ b\frac{\beta z_{a,2}+z_{A,2}}{z_{A,2}+z_{a,2}}-d-c(z_{A,2}+z_{a,2})-p\frac{z_{A,2}}{z_{A,2}+z_{a,2}}\Bigr]+p\frac{z_{A,1}z_{a,1}}{z_{A,1}+z_{a,1}}=0.
\end{align}
By subtracting \eqref{equa1} and \eqref{equa2}, and \eqref{equa3} and \eqref{equa4} we get
$$
(z_{A,i}-z_{a,i})\Bigl( b\beta-d-c(z_{A,i}+z_{a,i})\Bigr)=0, \quad i \in \mathcal{I}.
$$ 
Therefore equilibria are defined by the four following cases:
$$
\left\{
\begin{aligned}
&z_{A,1}=z_{a,1}\\
 &\text{or}\\
 & z_{A,1}+z_{a,1}=(b\beta-d)/c
\end{aligned}
\right.
\quad \text{and}\quad
\left\{
\begin{aligned}
 &z_{A,2}=z_{a,2}\\
& \text{or}\\
 & z_{A,2}+z_{a,2}=(b\beta-d)/c.
\end{aligned} 
\right.
$$
\underline{\textbf{1st case}}: $z_{A,1}=z_{a,1}$ and $z_{A,2}=z_{a,2}$.\\
From \eqref{equa1} and \eqref{equa3} we derive
$$z_{A,1}\Bigl[ b\frac{(\beta +1)}{2}-d-2cz_{A,1}-\frac{p}{2}\Bigr]=-\frac{z_{A,2}p}{2},$$
{and}
$$-\frac{z_{A,1}p}{2}=z_{A,2}\Bigl[ b\frac{(\beta +1)}{2}-d-2cz_{A,2}-\frac{p}{2}\Bigr].$$
By summing, we get $P(z_{A,1})=P(z_{A,2})$ where $P$ is the polynomial function defined by:
$$
P(X)=X\Bigl[ b\frac{(\beta +1)}{2}-d-p\Bigr]-2cX^2,
$$
whose roots are $0$ and 
$$\frac{b(\beta +1)-2d-2p}{4c}.$$ 
Then, either $z_{A,1}=z_{A,2}$ or $z_{A,1}$ and $z_{A,2}$ are symmetrical with respect to the maximum of $P$ which leads 
to 
$$z_{A,1}=\frac{b(\beta +1)-2d-2p}{4c}-z_{A,2}.$$ 
In the first case $z_{A,1}=z_{A,2}$, Equation \eqref{equa1} implies that either $z_{A,1}=0$, which gives the null equilibrium or 
$$z_{A,1}=\frac{b(\beta +1)-2d}{4c},$$ 
which gives equilibrium \eqref{eq-zzzz}. 
In the second case, we inject the expression of $z_{A,2}$ in \eqref{equa1} to obtain that $z_{A,1}$ satisfies:
$$
-2cX^2+AX+\frac{p}{4c}A=0,
$$
with $A= b(\beta +1)/2-d-p$. The discriminant of this degree $2$ equation is $A(A+2p)$. Therefore, either
$$
z_{A,1}=\frac{A+\sqrt{A(A+2p)}}{4c} \quad \text{and}\quad z_{A,2}=\frac{A-\sqrt{A(A+2p)}}{4c},
$$
or
$$
z_{A,1}=\frac{A-\sqrt{A(A+2p)}}{4c} \quad \text{and}\quad z_{A,2}=\frac{A+\sqrt{A(A+2p)}}{4c}.
$$
However, these equilibria are not positive.
\medskip\\
\underline{\textbf{2nd case}} : $ z_{A,1}+z_{a,1}=(b\beta-d)/c=\zeta=z_{A,2}+z_{a,2}$.\\
As previously, we obtain
$$
(b(\beta-1)+p)z_{A,1}\Bigl(\frac{z_{A,1}}{\zeta}-1\Bigr)=pz_{A,2}\Bigl(\frac{z_{A,2}}{\zeta}-1\Bigr),
$$
and
$$pz_{A,1}\Bigl(\frac{z_{A,1}}{\zeta}-1\Bigr)=(b(\beta-1)+p)z_{A,2}\Bigl(\frac{z_{A,2}}{\zeta}-1\Bigr).
$$
By summing these equalities, we get $Q(z_{A,1})=Q(z_{A,2})$ with
$$
Q(X)=X\Bigl(\frac{X}{\zeta}-1\Bigr)\bigl(b(\beta-1)+2p\bigr).
$$
Then, either $z_{A,1}=z_{A,2}$ and \eqref{equa1} gives that
$$
z_{A,1}\Bigl(\frac{z_{A,1}}{\zeta}-1\Bigr)=0,
$$
which gives equilibrium \eqref{eq-z0z0}, or $z_{A,1}=\zeta-z_{A,2}$ which implies $
z_{A,1}(z_{A,1}/\zeta-1)=0
$ and gives equilibrium \eqref{eq-z00z}.
\medskip\\
\underline{\textbf{3rd case}} : $ z_{A,1}=z_{a,1}, \text{ and } z_{A,2}+z_{a,2}=(b\beta-d)/c=\zeta$.\\
Substituting in Equations \eqref{equa1} and \eqref{equa4} we get that
$$z_{A,1}\Bigl[ b\frac{\beta+1}{2}-d-2c z_{A,1}-\frac{p}{2}\Bigr]+p\frac{z_{A,2}(\zeta-z_{A,2})}{\zeta}=0, $$
and
$$(\zeta- z_{A,2})\Bigl[ \dfrac{b}{\zeta}(\beta \zeta+(1-\beta) z_2^A)-d-c \zeta-p\frac{z_{A,2}}{\zeta}\Bigr]+p\frac{z_{A,1}}{2}=0.$$
Therefore, since $\zeta=(b\beta-d)/c$, these equations become
\begin{equation}
\label{eqcas3_1}
z_{A,1}=\dfrac{2}{p}(z_{A,2}-\zeta) z_{A,2} \Bigl[ \dfrac{b(1-\beta)-p}{\zeta}\Bigr],\end{equation}
and
\begin{equation*}
\label{eqcas3_2}
 \dfrac{(z_{A,2}-\zeta) z_{A,2}}{\zeta} \Big\{ \dfrac{2}{p} \left[ b(1-\beta)-p\right] \Bigl[ b\frac{\beta+1}{2}-d-\frac{p}{2} -\dfrac{4c}{p} (z_{A,2}-\zeta) z_{A,2} \dfrac{b(1-\beta)-p}{\zeta} \Bigr] -{p}\Big\}=0.
\end{equation*}
This last equation provides the following possible cases:
\begin{itemize}
\item $z_{A,2}=0$, which implies $z_{a,2}=\zeta $, and from \eqref{eqcas3_1} $z_{A,1}=z_{a,1}=0$ (Equilibrium \eqref{eq-z000}),
\item $z_{A,2}=\zeta$, which implies $z_{a,2}=0 $, and from \eqref{eqcas3_1} $z_{A,1}=z_{a,1}=0$ (Equilibrium \eqref{eq-z000}),
\item $z_{A,2}$ solution of
\begin{equation*}
(b(1-\beta)-p) \Big[ b\frac{\beta+1}{2}-d-\frac{p}{2} -\dfrac{4c}{p} (z_{A,2}-\zeta) z_{A,2} \dfrac{b(1-\beta)-p}{\zeta} \Big]-
\dfrac{p^2}{2}=0,
\end{equation*}
which can be summarized as 
\begin{equation} \label{eqcas3_seconddegre} (z_{A,2}-\zeta)z_{A,2}+C=0,
\end{equation}
where 
$$C=\dfrac{p\zeta}{8c(b(\beta-1)+p)^2} \left[ b^2(\beta^2-1)+2p(b-d)-2bd(\beta-1) \right].$$
The discriminant $\Delta$ of the degree $2$ Equation \eqref{eqcas3_seconddegre}
was introduced in Equation \eqref{defzbar}.
 A simple computation gives the sign of $\Delta$:
\begin{equation}
\begin{aligned}\label{Deltapositif}
\Delta
&={\zeta^2-4C}\\
&= \zeta^2-\dfrac{p\zeta}{2c(b(\beta-1)+p)^2} \left[ b^2(\beta^2-1)+2p(b-d)-2bd(\beta-1) \right]  \\
 &= \dfrac{\zeta}{2c(b(\beta-1)+p)^2} \Bigl[ 2b^2(\beta-1)^2(b\beta-d)\\
 &\hspace{4cm}+2bp(\beta-1)[b\beta-d+p]+b^2(\beta-1)^2p \Bigr]
  > 0.
\end{aligned}
\end{equation}
Thus \eqref{eqcas3_seconddegre} has two distinct solutions:
\begin{equation*}
z_{A,2}^{+}=\dfrac{\zeta+\sqrt{\Delta}}{2}>0 \quad \text{and} \quad  z_{A,2}^{-}=\dfrac{\zeta-\sqrt{\Delta}}{2}.
\end{equation*}
Since $C>0$, both roots $z_{A,2}^{-}$ and $z_{A,2}^{+}$ are strictly positive.\\
We finally deduce from \eqref{eqcas3_1} and \eqref{eqcas3_seconddegre} that in both cases $z_{A,2}=z_{A,2}^{-}$ and 
$z_{A,2}= z_{A,2}^{+}$ then
\begin{equation*}
 z_{A,1}=z_{a,1}=\dfrac{b^2(\beta^2-1)+2p(b-d)-2bd(\beta-1)}{4c(b(\beta-1)+p)}.
\end{equation*}
This gives equilibrium \eqref{eq-complique}, by symmetry between patches 1 and 2.
\end{itemize}
\bigskip

The end of this subsection provides a detailed exposition of the stability of fixed points of \eqref{systdet}.
 We consider separately each equilibrium and use symmetries of the dynamical system between patches 1 and 2 and between alleles $A$ and $a$.\smallskip\\
\underline{\textbf{Equilibrium \eqref{eq-z000}}}: 
By subtracting \eqref{equa2} from \eqref{equa1}, we obtain:
\begin{equation}\label{diffpop2}
\frac{d}{dt}(z_{A,1}-z_{a,1})=(z_{A,1}-z_{a,1})\Bigl( b\beta-d-c(z_{A,1}+z_{a,1})\Bigr).
\end{equation}
This equation provides the asymptotic instability since for this equilibrium, $z_{A,1}+z_{a,1}=0$.
\medskip\\
\underline{\textbf{Equilibrium \eqref{eq-z00z}}}: We consider the equilibrium $(\zeta,0,0,\zeta)$. The Jacobian matrix of the dynamical system at this fixed point is:
\begin{equation*}
 \begin{pmatrix}
  -(b\beta-d)     & b(1-2\beta)+d-p  	& p                &        0           \\
  0                & b(1-\beta)-p    	& p                &        0           \\
  0                &     p         	&  b(1-\beta)-p   &          0         \\
  0                &     p       	& b(1-2\beta)+d-p & -(b\beta-d)
 \end{pmatrix}
\end{equation*}
The eigenvalues are:
 $-b(\beta-1)$, $-b(\beta-1)-2p$, and $-(b\beta-d)$. All of them are negative, and $-(b\beta-d)$ is of multiplicity two. The equilibrium is therefore asymptotically stable.
 \medskip\\
\underline{\textbf{Equilibrium \eqref{eq-z0z0}}}: We consider the equilibrium $(0,\zeta,0,\zeta)$. The Jacobian matrix of the dynamical system at this fixed point is:
\begin{equation*}
 \begin{pmatrix}
  b(1-\beta)-p    & 0            & p                &        0           \\
  b(1-2\beta)+d-p & -(b\beta-d) & p                &        0           \\
  p                &     0        &b(1-\beta)-p     &          0         \\
  p                &     0        & b(1-2\beta)+d-p & -(b\beta-d)
 \end{pmatrix}
\end{equation*}
The eigenvalues are:
 $-b(\beta-1)$, $-b(\beta-1)-2p$, and $-(b\beta-d)$. All of them are negative, and $-(b\beta-d)$ is of multiplicity two.
The equilibrium is therefore asymptotically stable.
\medskip\\
\underline{\textbf{Equilibrium \eqref{eq-zzzz}}}: The Jacobian matrix of the dynamical system at this fixed point is:
\begin{equation*}
 \dfrac{1}{4} 
  \begin{pmatrix}
  2(d-b)-p         & 2(d-b\beta)-p            	& p                &        p           \\
  2(d-b\beta)-p   &    2(d-b)-p         	& p                &        p           \\
  p                &     p                      &  2(d-b)-p        & 2(d-b\beta)-p         \\
  p                &     p        		& 2(d-b\beta)-p   &    2(d-b)-p
 \end{pmatrix}
\end{equation*}
The eigenvalues are: $-(b(\beta+1)/2-d)$, $-(b(\beta+1)/2-d+p)$ and $b(\beta-1)/2$. 
$-(b(\beta+1)/2-d)$ and $-(b(\beta+1)/2-d+p)$ are negative, and $b(\beta-1)/2$ is positive and of multiplicity two.
The equilibrium is thus unstable.
\medskip\\
\underline{\textbf{Equilibrium \eqref{eq-complique}}}: 
Recall the definition of $\wt{\zeta}$ in \eqref{defzbar}
and assume that $z_{A,1}=z_{a,1}=\wt{\zeta}$.  We first prove that at this fixed point, 
\begin{equation} \label{2tildezinfz} z_{A,1}+z_{a,1}= 2 \wt{\zeta} <\zeta,\end{equation} 
which is equivalent to
$$ b^2(\beta^2-1)+2p(b-d)-2bd(\beta-1)<2(b(\beta-1)+p)(b\beta-d). $$
A straightforward computation leads to 
$$b^2(\beta^2-1)+2p(b-d)-2bd(\beta-1)-2(b(\beta-1)+p)(b\beta-d)=-b(\beta-1)(2p+b(\beta-1)),$$
which is negative and thus proves the inequality. 
From \eqref{2tildezinfz} we deduce that near the equilibrium \eqref{eq-complique}, $b\beta-d-c(z_{A,1}+z_{a,1})>0$. 
The instability then derives from Equation \eqref{diffpop2}.

\subsection{Fixed points and stability when $\beta=1$} \label{sectionfixedpoints1}
\label{sectionbeta1}
Following a similar reasoning to the one in Section \ref{sectionfixedpoints}, we obtain that 
the equilibria of the system are exactly the lines 
$ \mathcal{L}$ and $ \tilde{\mathcal{L}} $ defined in Theorem \ref{thm_systdet}. 
A study of the Jacobian matrices proves that these equilibria are no longer hyperbolic. 
It ends the proof of Theorem \ref{thm_systdet}.

\subsection{Proof of Proposition \ref{prop1}}

This subsection is devoted to the proof of Proposition~\ref{prop1}.
The idea is to find a solution of the form 
$$\psi(t)=\gamma(t)\mathbf{v}(w,x) \quad  \text{with} \quad \gamma(0)=1,$$
where $\mathbf{v}(w,x)=(w-x,x,x,w-x)$ has been introduced in Proposition~\ref{prop1}.
Assuming that $\psi$ is solution to the system \eqref{systdet} with $\beta=1$, we deduce that for all $(\alpha,i)\in\mathcal{E}$:
\begin{equation*}\begin{aligned}
\frac{d}{dt}\psi_{\alpha,i}(t)&=\frac{d}{dt}\gamma(t) v_{\alpha,i} (w,x)\\
&=\psi_{\alpha,i}(t)(b-d-c(\psi_{\alpha,i}(t)+\psi_{\bar{\alpha},i}(t))) +p\frac{\psi_{\alpha,i}(t)\psi_{\bar{\alpha},i}(t)}{\psi_{\alpha,i}(t)+\psi_{\bar{\alpha},i}(t)}-p\frac{\psi_{\alpha,\bar{i}}(t)\psi_{\bar{\alpha},\bar{i},}(t)}{\psi_{\alpha,\bar{i}}(t)+\psi_{\bar{\alpha},\bar{i},}(t)}\\
&=\gamma(t)v_{\alpha,i} (w,x)(b-d-cw\gamma(t)).
\end{aligned}\end{equation*}
Thus $\gamma(t)$ satisfies the logistic equation 
$$\frac{d}{dt}\gamma(t)=\gamma(t)(b-d-cw\gamma(t)),$$
whose solution starting from $1$ is given by
\begin{equation}\label{gamma}
\gamma(t)=\frac{e^{t(b-d)}}{1+\frac{cw}{b-d}(e^{t(b-d)}-1)}.\end{equation}
In particular $\gamma(t)$ converges to $(b-d)/cw=\zeta/w$ as $t\to\infty$.\\
A standard computation proves that $\psi(t)=\gamma(t)\mathbf{v}(w,x)$ with $\gamma$ chosen according to \eqref{gamma} is the solution 
to \eqref{systdet} starting from $\mathbf{v}(w,x)$ and converges to $\zeta \mathbf{v}(w,x)/w= \mathbf{u}( \zeta x/w)$. This ends the proof of Proposition~\ref{prop1}.

\subsection{Containment and Lyapunov function for a small migration rate}

In this subsection, we are mainly interested in Equilibrium \eqref{eq-z00z}. 
Recall the definition of $\mathcal{D}$ in \eqref{defdelta}
\begin{equation*} \mathcal{D}:=\{ z\in \R_+^\mathcal{E}, z_{A,1}-z_{a,1}>0, z_{a,2}-z_{A,2}>0 \}.\end{equation*}
First, we prove that we can restrict our attention to the bounded set $\mathcal{K}_p\subset\mathcal{D}$ 
defined in \eqref{defK}. For the sake of readability, we introduce the two real numbers
\begin{equation} \label{defzminzmax} z_{min}:= \frac{b(\beta+1)-2d-p}{2c}\leq \zeta \leq \zeta+\frac{p}{2c}=:z_{max}, \end{equation}
which allows one to write the set $ \mathcal{K}_p$ defined in \eqref{defK} as
$$  \mathcal{K}_p:= \left\{ 
\tb{z} \in \mathcal{D}, \; \{ z_{A,1}+z_{a,1}, \ z_{A,2}+z_{a,2} \} \in \left[ z_{min}, z_{max}\right] 
\right\}. $$

\begin{lem}
\label{lemma_invariantset}
Assume that $p<b(\beta+1)-2d$. The set $\mathcal{K}_p$
is invariant under the dynamical system~\eqref{systdet}.
Moreover, any solution to~\eqref{systdet} starting from the set 
$ \mathcal{D} $
reaches $\mathcal{K}_p$ after a finite time.
\end{lem}

\begin{proof}
First, Equation~\eqref{diffpop2} and the symmetrical equation for the patch 2 are sufficient to prove that the subset
$ \mathcal{D}$
is invariant under the dynamical system.

Second, we prove that the trajectory reaches the bounded set $\mathcal{K}_p$ in a finite time and third that $\mathcal{K}_p$ is stable.
The dynamics of the total population size $n=z_{A,1}+z_{a,1}+z_{A,2}+z_{a,2}$ satisfies
$$
\frac{dn}{dt}=n(\beta b-d)-2b(\beta-1)\left( \frac{z_{A,1}z_{a,1}}{z_{A,1}+z_{a,1}} + \frac{z_{A,2}z_{a,2}}{z_{A,2}+z_{a,2}} \right)-c((z_{A,1}+z_{a,1})^2+(z_{A,2}+z_{a,2})^2).
$$
Since $(a+b)^2\leq 2(a^2+b^2)$ for every real numbers $(a,b)$,
$$
\frac{dn}{dt}\leq n \left( \beta b -d - \frac{c}{2}n \right).
$$
Using classical results on logistic equations, we deduce that 
\begin{equation}
\label{lim_suptotal}
\limsup_{t\to +\infty} n(t)\leq 2\zeta.
\end{equation}
Let $\eps$ be positive, and suppose that for any $t>0$,
$ (z_{A,1}+z_{a,1})(t)\leq \zeta-\eps, $
then 
using \eqref{diffpop2} we have for $t \geq 0$,
\begin{equation}
\label{eq_diffinfty}
 z_{A,1}(t)\geq (z_{A,1}-z_{a,1})(t)\geq (z_{A,1}-z_{a,1})(0)e^{c\eps t} \underset{t\to+\infty}{\to} +\infty.
\end{equation}
This contradicts~\eqref{lim_suptotal}. As a consequence,  
\begin{equation}\label{audessusde}
           \exists \ t <\infty,\quad   (z_{A,1}+z_{a,1})(t) \geq \zeta-\eps.
          \end{equation}
In particular, this result holds for $\zeta-\eps_0=z_{min}$ where $\eps_0=(p+b(\beta-1))/2c$.\\
Furthermore, the dynamics of the total population size in the patch $1$ satisfies the following equation:
\begin{equation}
\label{sumpop1}
\begin{aligned}
\frac{d}{dt}(z_{A,1}+z_{a,1})= &(z_{A,1}+z_{a,1})(b \beta - d - c (z_{A,1}+z_{a,1}))\\
&-2(b(\beta-1)+p)\frac{z_{A,1}z_{a,1}}{z_{A,1}+z_{a,1}}+2p\frac{z_{A,2}z_{a,2}}{z_{A,2}+z_{a,2}}. 
\end{aligned}
\end{equation}
By noticing that $z_{A,1}z_{a,1} \leq (z_{A,1}+z_{a,1})^2/4$, we get 
\begin{equation} \label{minpopsizepatch}
\begin{aligned}
 \frac{d}{dt}(z_{A,1}+z_{a,1})&\geq (z_{A,1}+z_{a,1})\left(b \beta-d-c (z_{A,1}+z_{a,1})\right)-
\left(b(\beta-1)+p\right)\frac{z_{A,1}+z_{a,1}}{2}\\
&\geq c(z_{A,1}+z_{a,1})\left(z_{min}- (z_{A,1}+z_{a,1})\right).
\end{aligned}
\end{equation}
The last term becomes positive as soon as $z_{A,1}+z_{a,1}\leq z_{min}$. As a consequence, once the total population size 
in the patch 1 is larger than $z_{min}$, it stays larger than this threshold.
Using symmetrical arguments, the same conclusion holds for the patch 2.
Using additionaly \eqref{audessusde}, we find $t_{min}>0$ such that $\forall t\ge t_{min}$, 
\begin{equation}
\label{K_min}
z_{A,i}(t)+z_{a,i}(t) \ge z_{min}, \quad \forall i\in  \mathcal{I}, \text{ and } n(t)\leq 2\zeta +1.
\end{equation}
We now focus on the upper bound of the set $\mathcal{K}_p$ by bounding from above the total population size in the patch $i$, for all 
$t\geq t_{min}$,
\begin{equation}
\label{eq_ingsumbis}
\begin{aligned}
\frac{d}{dt}(z_{A,i}+z_{a,i})
&\leq (2\zeta+1)(c\zeta -  c(z_{A,i}+z_{a,i}))+\frac{p}{2} (2\zeta+1)\\
&\leq c(2\zeta+1)\left(z_{max}-(z_{A,i}+z_{a,i})\right).
\end{aligned}
\end{equation}
This implies that, if $\alpha>0$ is fixed, there exists $t_\alpha\geq t_{min}$ such that $z_{A,i}(t)+z_{a,i}(t)\leq z_{max}+\alpha$ 
for all $i \in \mathcal{I}$ and $t\geq t_{\alpha}$.\\
Finally, we use a proof by contradiction to ensure that the trajectory hits the compact $\mathcal{K}_p$. Let us assume that for any $t\geq t_\alpha$, 
\begin{equation}
\label{eq_absurdebis}
z_{A,1}(t)+z_{a,1}(t)\geq z_{max}-\alpha.
\end{equation}
From~\eqref{diffpop2}, 
and choosing an $\alpha<p/2c$,
we deduce that $z_{A,1}-z_{a,1}$ converges to $0$. In addition with~\eqref{eq_absurdebis}, we find $t'_{\alpha}\geq t_\alpha$ such that for any $t\geq t'_{\alpha}$,
\begin{equation}
\label{eq_quotientbis}
 \frac{z_{A,1}(t)z_{a,1}(t)}{z_{A,1}(t)+z_{a,1}(t)} \geq \frac{1}{4}\left( z_{max}-2\alpha \right). 
\end{equation}
We insert \eqref{eq_quotientbis} in the equation~\eqref{sumpop1} to deduce that, for all $t\geq t'_{\alpha}$,
\begin{equation*}
\begin{aligned}
\frac{d}{dt}&(z_{A,1}+z_{a,1})\\ &\leq  c\left(2\zeta+1\right)(\zeta -  (z_{A,1}+z_{a,1}))
 -\frac{b(\beta-1)+p}{2} \left( z_{max}-2\alpha \right)+\frac{p}{2}  \left(2\zeta+1\right). \\
&\leq c\left(2\zeta+1\right)\left(z_{max}-2\alpha-  (z_{A,1}+z_{a,1})\right) +2\alpha c(2\zeta+1)
 -\frac{b(\beta-1)+p}{2} \left( z_{max}-2\alpha \right).
\end{aligned}
\end{equation*}
The first term of the last line is negative under Assumption~\eqref{eq_absurdebis}, thus, if $\alpha$ is sufficiently small,
\begin{equation}
\begin{aligned}
\frac{d}{dt}(z_{A,1}+z_{a,1})
& \leq  -\frac{1}{2} \left[ (b(\beta-1)+p)z_{max}\right]
+\alpha \left[  b(\beta-1) +\frac{p}{2c}(2\zeta+1)  \right]\\
&\leq  -\frac{1}{4} \left[ (b(\beta-1)+p)z_{max} \right].
\end{aligned}
\end{equation}
This contradicts~\eqref{eq_absurdebis}. Thus, the total population size of the patch $1$ is lower than 
$z_{max}-\alpha$ after a finite time. Moreover, \eqref{eq_ingsumbis} ensures that once the total population size 
of the patch $1$ has reached the threshold 
$z_{max}$, it stays smaller than this threshold. Reasoning similarly for the patch $2$, 
we finally find a finite time such that the trajectory hits the compact $\mathcal{K}_p$ and remains in it afterwards.
This ends the proof of Lemma~\ref{lemma_invariantset}.
\end{proof}

As $\mathcal{D}$ is invariant under the dynamical system \eqref{systdet}, we can
 consider the function $V: \mathcal{D} \to \R$:
\begin{equation}\label{defV} V(\tb{z})= \ln \left(\frac{z_{A,1}+z_{a,1}}{z_{A,1}-z_{a,1}}\right)
 +\ln \left(\frac{z_{a,2}+z_{A,2}}{z_{a,2}-z_{A,2}}\right) . \end{equation}
It characterizes the dynamics of~\eqref{systdet} on $\mathcal{K}_p$. Indeed, as proved in the next lemma, 
V is a Lyapunov function 
if $p$ is sufficiently small.
This will allow us to prove that the solutions to \eqref{systdet} converge to $(\zeta,0,0,\zeta)$ exponentially 
fast as soon as their trajectory hits the set $\mathcal{K}_p$.
Before stating the next lemma, we introduce the positive real number:
\begin{equation}\label{defC1}
C_1:= \frac{1}{2}\left( \frac{2b(\beta-1)+2p}{z_{min}}- \frac{2p}{z_{max}} \right),
\end{equation}
where $z_{min}$ and $z_{max}$ have been defined in \eqref{defzminzmax}.
Then we have the following result:
\begin{lem}
\label{lemma_V}
Assume that $p<p_0$ defined in \eqref{defp0}. Then
$V(\tb{z}(t))$ is non-negative and non-increasing on  
$\mathcal{K}_p$, and satisfies
\begin{equation} \label{majdotV}  \frac{d}{dt}V(\tb{z}(t)) \leq - C_1 (z_{a,1}(t)+z_{A,2}(t)), \quad t \geq 0. \end{equation}
\end{lem}

\begin{proof}
 For $i\in \mathcal{I}$ and $\tb{z}\in \mathcal{K}_p$, $z_{\alpha_i,i}-z_{\bar\alpha_i,i}\leq z_{\alpha_i,i}+z_{\bar\alpha_i,i}$, where $\alpha_1=A, \alpha_2=a$ and $\bar\alpha_i=\mathcal{A} \setminus \alpha_i$. 
 Thus, $V(\tb{z})\geq 0$.
Now,
\begin{eqnarray} \label{deriveeV}
 \frac{d}{dt}V(\tb{z}(t))
&=& \frac{\dot{z}_{A,1}(t)+\dot{z}_{a,1}(t)}{z_{A,1}(t)+z_{a,1}(t)} 
 -\frac{\dot{z}_{A,1}(t)-\dot{z}_{a,1}(t)}{z_{A,1}(t)-z_{a,1}(t)}
 + \frac{\dot{z}_{A,2}(t)+\dot{z}_{a,2}(t)}{z_{A,2}(t)+z_{a,2}(t)} 
 -\frac{\dot{z}_{a,2}(t)-\dot{z}_{A,2}(t)}{z_{a,2}(t)-z_{A,2}(t)}\nonumber \\
 &=& -\underset{i=1,2}{\sum}
   \frac{z_{A,i}z_{a,i}}{z_{A,i}+z_{a,i}}\left[ \frac{2b(\beta-1)+2p}{z_{A,i}+z_{a,i}}- \frac{2p}{z_{A,\bar i}+z_{a,\bar i}} \right],
\end{eqnarray}
from \eqref{diffpop2} and \eqref{sumpop1}.
Thus, $dV(\tb{z}(t))/dt$ is nonpositive if 
\begin{equation}
\label{cond_p}
\frac{b(\beta-1)}{p} > \max\left\{\frac{z_{A,1}+z_{a,1}}{z_{A,2}+z_{a,2}}-1, \frac{z_{A,2}+z_{a,2}}{z_{A,1}+z_{a,1}}-1\right\}.
\end{equation}
Since $\tb{z}$ belongs to $\mathcal{K}_p$, the r.h.s of~\eqref{cond_p} can be bounded from above by
$$
\frac{z_{max}}{z_{min}}-1= \frac{b(\beta-1)+2p}{b(\beta+1)-2d-p}.
$$
Therefore, the condition~\eqref{cond_p} is satisfied if  
$$\frac{b(\beta-1)}{p}> \frac{b(\beta-1)+2p}{b(\beta+1)-2d-p},$$ 
that is, if 
$$p<\frac{\sqrt{b(\beta-1)[b(3\beta+1)-4d]}-b(\beta-1)}{2}=p_0,$$
and under this condition, 
$$ \frac{2b(\beta-1)+2p}{z_{A,i}+z_{a,i}}- \frac{2p}{z_{A,\bar i}+z_{a,\bar i}}\geq 2C_1, \quad z \in \mathcal{K}_p, \quad i \in \mathcal{I}. $$
Moreover, as the set $\mathcal{D}$ is invariant under the dynamical system \eqref{systdet}, $z_{A,1}$ stays larger that $z_{a,1}$, and 
$$ \frac{z_{A,1}}{z_{A,1}+z_{a,1}}\geq \frac{1}{2}. $$
In the same way,
$$ \frac{z_{a,2}}{z_{A,2}+z_{a,2}}\geq \frac{1}{2}. $$
As a consequence, the first derivative of $V$ satisfies \eqref{majdotV} for every $t \geq 0$.
\end{proof}
We now have all the ingredients to prove Theorem \ref{theoCvceD}.

\subsection{Proof of Theorem \ref{theoCvceD}}

Lemma \ref{lemma_invariantset} states that 
 any solution to~\eqref{systdet} starting from the set 
$ \mathcal{D} $
reaches $\mathcal{K}_p$ after a finite time.
Let us show that because of Lemma~\ref{lemma_V}, any solution to~\eqref{systdet} which starts from $\mathcal{K}_p$ converges 
exponentially fast
to $(\zeta,0,0,\zeta)$ when $t$ tends to infinity. To do this, we need to introduce some positive constants
$$ C_2:= z_{min}^2 e^{-V(\tb{z}(0))} , \quad C_3:= \frac{2}{C_2}z_{max} $$
$$ C_4:= \frac{z_{max}}{2}V(\tb{z}(0)) ,  \quad  C_5:= z(4b\beta-2d+3p) C_4 ,  $$
where we recall that $z_{min}$ and $z_{max}$ have been defined in \eqref{defzminzmax}.

First, we prove that the population density differences $z_{A,1}-z_{a,1}$ and $z_{a,2}-z_{A,2}$ cannot be too small. To do this, we use the decay of the function $V$ 
stated in Lemma \ref{lemma_V}: 
\begin{eqnarray*}
  V(\tb{z}(0))\geq V(\tb{z}(t))&= &\ln \left(\frac{z_{A,1}(t)+z_{a,1}(t)}{z_{A,1}(t)-z_{a,1}(t)}\frac{z_{a,2}(t)+z_{A,2}(t)}{z_{a,2}(t)-z_{A,2}(t)}\right)\\
 & \geq & \ln \left(\frac{z_{min}^2}{(z_{A,1}(t)-z_{a,1}(t))(z_{a,2}(t)-z_{A,2}(t))}\right).
\end{eqnarray*}
This implies that 
\begin{equation} \label{mindiffpop} (z_{A,1}(t)-z_{a,1}(t))(z_{a,2}(t)-z_{A,2}(t))\geq C_2. \end{equation}
Now, from the inequality $\ln x \leq x-1$ for $x \geq 1$ we deduce for $\tb{z}$ in $\mathcal{K}_p$,
\begin{multline} \label{majV}
 V(\tb{z}) \leq  \left(\frac{z_{A,1}+z_{a,1}}{z_{A,1}-z_{a,1}}-1  \right) + \left( \frac{z_{a,2}+z_{A,2}}{z_{a,2}-z_{A,2}}-1 \right)  \\
 = 2\frac{z_{a,1}(z_{a,2}-z_{A,2})+z_{A,2}(z_{A,1}-z_{a,1})}{(z_{A,1}-z_{a,1})(z_{a,2}-z_{A,2})} 
 \leq C_3 (z_{a,1}+z_{A,2}) ,
 \end{multline}
where we have used that $z\in \mathcal{K}_p$ and inequality \eqref{mindiffpop}. Then combining \eqref{majdotV} and \eqref{majV}, we get 
\begin{equation}
  \frac{d}{dt}V(\tb{z}(t)) \leq - \frac{C_1}{C_3}V(\tb{z}(t)) ,
\end{equation}
which implies for every $t \geq 0$:
\begin{equation}
 V(\tb{z}(t)) \leq V(\tb{z}(0))e^{-C_1 t /C_3}.
\end{equation}
Now, from the inequality $\ln x \geq (x-1)/x$ for $x \geq 1$ we deduce for $\tb{z}$ in $\mathcal{K}_p$,
\begin{equation}\label{minV}
\begin{aligned}
 V(\tb{z}) \geq  \left(\frac{z_{A,1}+z_{a,1}}{z_{A,1}-z_{a,1}}-1  \right)\frac{z_{A,1}-z_{a,1}}{z_{A,1}+z_{a,1}} &+ 
 \left( \frac{z_{a,2}+z_{A,2}}{z_{a,2}-z_{A,2}}-1 \right) \frac{z_{a,2}-z_{A,2}}{z_{a,2}+z_{A,2}} \\
= & \frac{2z_{a,1}}{z_{A,1}+z_{a,1}} + \frac{2z_{A,2}}{z_{a,2}+z_{A,2}} \geq \frac{2}{z_{max}}(z_{a,1}+z_{A,2}).
\end{aligned}
\end{equation}
Hence, 
\begin{equation}
 z_{a,1}(t)+z_{A,2}(t) \leq  C_4 e^{-C_1 t /C_3},
\end{equation}
and the exponential convergence of $ z_{a,1}$ and $z_{A,2}$ to $0$ is proved. Let us now focus on the two other variables, $ z_{A,1}$ and $z_{a,2}$.
From the definition of the dynamical system in \eqref{systdet}, and noticing that $| z_{A,1}(t)- \zeta|\leq \zeta$ as $z\in \mathcal{K}_p$, we get
\begin{equation*}
\begin{aligned}
 \frac{d}{dt} \left( z_{A,1}(t)- \zeta \right)^2 =& - 2c z_{A,1}(t) \left( z_{A,1}(t)- \zeta \right)^2  + 2pz_{a,2}(t)( z_{A,1}(t)- \zeta )\frac{z_{A,2}(t)}{z_{A,2}(t)+z_{a,2}}\\
 & - 2z_{a,1}(t)( z_{A,1}(t)- \zeta )\left(c z_{A,1}(t)+(p+b(\beta - 1))\frac{z_{A,1}(t)}{z_{A,1}(t)+z_{a,1}(t)}  \right)\\
 \leq& -c z_{min} \left( z_{A,1}(t)- \zeta \right)^2 + 2p\zeta z_{A,2}(t)+2\zeta z_{a,1}(t)\left(c z_{max}+p+b(\beta - 1) \right)\\
\leq&  -c z_{min} \left( z_{A,1}(t)- \zeta \right)^2 
  + \zeta(4b\beta-2d+3p)( z_{a,1}(t)+ z_{A,2}(t) )\\
\leq&  -c z_{min} \left( z_{A,1}(t)- \zeta \right)^2 
  + C_5 e^{-C_1 t /C_3} .
\end{aligned}
\end{equation*}
Hence, a classical comparison of nonnegative solutions of ordinary differential equations yields
$$ ( z_{A,1}(t)- \zeta )^2 \leq \left(( z_{A,1}(0)- \zeta )^2 - \frac{C_5}{c z_{min}- C_1/C_3}  \right) e^{-c z_{min}t}+ \frac{C_5}{c z_{min}- C_1/C_3}e^{-C_1 t /C_3},  $$
which gives the exponential convergence of $z_{A,1}$ to $\zeta$.
Reasoning similarly for the term $ z_{a,2}$ ends the proof of Theorem \ref{theoCvceD}.

\section{Stochastic process} \label{sectionsto}

In this section, we study properties of the stochastic process $(\tb{N}^K(t),t\geq0)$. We 
derive an approximation for the extinction time of subpopulations under some small initial conditions, and then
combine the results of this section with these on dynamical system (Section \ref{sectiondet}) to prove Theorem \ref{maintheo}.

\subsection{Approximation of the extinction time}
Let us first study the stochastic system $(\tb{Z}^{K}(t), t \geq 0)$ around the equilibrium $(\zeta,0,0,\zeta)$ when $K$ is large. 
The aim is to estimate the time before the loss of all $a$-individuals in the patch $1$ and all $A$-individuals in the patch $2$, which we denote by
\begin{equation}
 T_0^K=\inf\{t\geq 0, Z^K_{a,1}(t)+Z^K_{A,2}(t)=0\}.
 \end{equation}
Recall that $\zeta=({b\beta-d})c^{-1}>0$ and that the sequence of initial states $(\tb{Z}^K(0),K\geq1)$ converges 
in probability when $K$ goes to infinity to a deterministic vector ${\bf{z}^0}=(z_{A,1}^0, z_{a,1}^0, z_{A,2}^0, z_{a,2}^0)\in \R_+^\mathcal{E}$.

 \begin{prop} \label{propexttime}
 There exist two positive constants $\varepsilon_0$ and $C_0$ such that for any $\varepsilon\leq \varepsilon_0$,
 if there exists $\eta\in ]0,1/2[$ such that $\max(|z_{A,1}^0-\zeta|,|z_{a,2}^0-\zeta|) \leq \eps$ and
 $\eta\eps/2 \leq z_{a,1}^0,z_{A,2}^0 \leq \eps/2$, then 
 $$
\begin{aligned}
 &\text{for any } C>(b(\beta-1))^{-1}+C_0\eps, &\P(T_0^K\leq C \log(K)) \underset{K\to +\infty}{\to} 1,\\
 &\text{for any } 0\leq C <(b(\beta-1))^{-1}-C_0\eps, & \P(T_0^K\leq C \log(K)) \underset{K\to +\infty}{\to} 0.
\end{aligned}
$$
\end{prop}
Remark that the upper bound on $T^K_0$ still holds if $z_{a,1}^0=0$ or $z_{A,2}^0=0$. Moreover, if $z_{a,1}^0=z_{A,2}^0=0$, then the upper bound is satisfied with $C_0= 0$.
In the case where $\eta=0$, the upper bound of the extinction time still holds but not the lower bound. Indeed, as the initial conditions $z_{a,1}^0$ and $z_{A,2}^0$ go to $0$, 
the extinction time is faster.

\begin{proof}The proof relies on several coupling arguments. Our first step is to prove that the population sizes 
$Z^{K}_{A,1}$ and $Z^{K}_{a,2}$ remain close to $\zeta$ on a long time scale. In a second step, we couple the processes 
$Z^{K}_{a,1}$ and $Z^{K}_{A,2}$ with subcritical branching processes whose extinction times are known. 
We begin with introducing some additional notations:  
for any $\gamma, \varepsilon>0$ and $(\alpha,i)\in \mathcal{E}$,
\begin{equation}
 \label{def_R}
 R^{K,\gamma}_{\alpha,i}=\inf \{ t\geq 0, |Z^K_{\alpha,i}(t)-\zeta|\geq \gamma\},
\end{equation}
and  
\begin{equation}
 \label{def_T}
 T^{K,\eps}_{\alpha,i}=\inf\{ t\geq 0, Z^{K}_{\alpha,i}(t)\geq \eps \}.
\end{equation}
\textbf{Step 1: }
The first step consists in proving that as long as the population processes $Z^K_{a,1}$ and $Z^K_{A,2}$ 
have small values, the processes $Z^K_{A,1}$ and $Z^K_{a,2}$ stay close to $\zeta$. 
To this aim, we study the system on the time interval
$$ I_1^{K,\eps}:= \left[0, R^{K,\zeta/2}_{A,1} \wedge R^{K,\zeta/2}_{a,2} \wedge T^{K,\eps}_{a,1} \wedge T^{K,\eps}_{A,2}\right] ,$$
where $a\wedge b$ stands for $\min(a,b)$.\\
Let us first bound the rates of the population process $Z^K_{A,1}$.
\begin{itemize}
 \item We start with the birth rate of $A$-individuals in the patch $1$. Let us remark that as $\beta>1$, 
 the ratio $(\beta x +y)/(x+y)\le \beta$ for any $x,y \in \R_+$. 
 Moreover, the function $x\mapsto (\beta x +y)/(x+y)$ increases with $x$, for any $y\in \R_+$. Combining these observations with the 
 fact that for any $t<T^{K,\varepsilon}_{a,1}\wedge R^{K, \zeta/2}_{A,1}$, 
 $0\le Z^{K}_{a,1}(t)\le \varepsilon$ and $Z^{K}_{A,1}(t)\ge \zeta/2$, we deduce that the birth rate of $A$-individuals 
 in the patch $1$, $K\wt{\lambda}_{A,1}(\tb{Z}^K(t))$, defined in \eqref{deftildeb} can be bounded:
\begin{equation*}
 b\beta \left(\frac{\zeta}{\zeta+2\eps}\right) K Z^K_{A,1}(t) \leq K\wt{\lambda}_{A,1}(\mathbf{Z}^K(t)) 
 \leq b\beta K Z^K_{A,1}(t).
\end{equation*}
 \item The migration rate of $A$-individuals from the patch $2$ to the patch $1$ is sandwiched as follows  
for any $t<T^{K,\varepsilon}_{a,1}\wedge R^{K,\zeta/2}_{A,1}$:
\begin{equation*}
 0 \leq K\wt{\rho}_{2\to1}(\tb{Z}^K(t))\leq Kp\eps.
\end{equation*}
 \item The death rate of $A$-individuals in the patch $1$ and the migration rate from patch $1$ to patch $2$ are bounded as follows:
\begin{equation*}
 (d+c Z^K_{A,1}(t))K Z^K_{A,1}(t) \leq K\wt{d}_{A,1}(\tb{Z}^K(t)) \leq (d+c\eps +cZ^K_{A,1}(t))K Z^K_{A,1}(t),
\end{equation*}
\begin{equation*}
 0 \leq K\wt{\rho}_{1 \to 2}(\tb{Z}^K(t))\leq K p\eps.
\end{equation*}
\end{itemize}
Hence, using an explicit construction of the process $Z^K_{A,1}$ by means of Poisson point measures as in \eqref{def_poisson}, 
we deduce that on the time interval $I_1^{K,\eps}$, $Z^{K}_{A,1}$ is stochastically bounded by 
\begin{equation*}
 \mathcal{Y}^K_{inf} \preccurlyeq Z^K_{A,1} \preccurlyeq \mathcal{Y}^K_{sup},
\end{equation*}
where $\mathcal{Y}^K_{inf}$ is a $\N / K$-valued Markov jump process with transition rates
\begin{equation*}
 \begin{aligned}
  Kb\beta \left( 1-\frac{2\eps}{\zeta+2\varepsilon}\right) \frac{i}{K} &\quad \text{ from } \frac{i}{K} \text{ to } \frac{(i+1)}{K},\\
  K\left(\left(d+c\eps +c \frac{i}{K} \right)\frac{i}{K} +p\eps \right) & \quad \text{ from } \frac{i}{K} \text{ to } \frac{(i-1)}{K},
 \end{aligned}
\end{equation*}
and initial value $Z^K_{A,1}(0)$, and $\mathcal{Y}^K_{sup}$ is a $\N / K$-valued Markov jump process with transition rates
\begin{equation*}
 \begin{aligned}
  K\left(b\beta \frac{i}{K} + p\eps\right) &\quad \text{ from } \frac{i}{K} \text{ to } \frac{(i+1)}{K},\\
  K\left(d +c \frac{i}{K} \right)\frac{i}{K}  & \quad \text{ from } \frac{i}{K} \text{ to } \frac{(i-1)}{K}.
 \end{aligned}
\end{equation*}
and initial value $Z^K_{A,1}(0)$.\\
Let us focus on the process $\mathcal{Y}^K_{inf}$. Using a proof similar to the one of Lemma~\ref{lemapprox}, 
we prove that since the sequence $(\mathcal{Y}^K_{inf}(0),K\geq 1)$ converges in probability to the deterministic value $z_{A,1}^0$, 
$$\underset{K\to +\infty}{\lim} \underset{s\leq t}{\sup} \ |\mathcal{Y}^K_{inf}(s)-\Phi_{inf}(s)|=0\quad\quad a.s$$
for every finite time $t>0$,
where $\Phi_{inf}$ is the solution to
\begin{equation}
\label{eq_Phiinf}
\Phi'(t)=b\beta (1-2\eps/(\zeta+2\varepsilon)) \Phi(t)-p\eps-(d+c\eps+c\Phi(t))\Phi(t)
\end{equation}
with initial value $z_{A,1}^0$. Let us study the trajectory of $\Phi_{inf}$. The polynomial in $\Phi(t)$ on the r.h.s. of~\eqref{eq_Phiinf} has two roots
\begin{equation}
\label{def_phiinf+-}
\begin{aligned}
 \Phi^{\pm}_{inf}&=\frac{1}{2c}\left(b\beta \left(1-\frac{2\eps}{\zeta+2\varepsilon}\right)-d-c\eps \pm \sqrt{\left(b\beta \left(1-\frac{2\eps}{\zeta+2\varepsilon}\right)-d-c\eps\right)^2-4pc\eps} \right)\\
&= \frac{\zeta}{2}-\frac{\eps}{2}  \left(\frac{2b\beta}{(\zeta+2\varepsilon)c}+1\right) \pm \sqrt{\left(\frac{\zeta}{2}- \frac{\eps}{2}  \left(\frac{2b\beta}{(\zeta+2\varepsilon)c}+1\right) \right)^2-\frac{p\eps}{c}}.
\end{aligned}
\end{equation}
As a consequence, $\Phi'>0$ if and only if $\Phi \in ] \Phi^{-}_{inf}, \Phi^{+}_{inf}[$.
Definition \eqref{def_phiinf+-} implies that for small $\eps$,
$$ \Phi^{-}_{inf}\sim pc\eps. $$
Hence, if
$\eps_0$ is chosen sufficiently small and for any $\eps<\eps_0$, 
$$ \Phi^{-}_{inf}\leq 2pc\eps_0 <z^0_{A,1}. $$
Thus, we observe that any solution to~\eqref{eq_Phiinf} with initial condition $\Phi_{inf}(0)\in [2pc\eps_0,+\infty[$ is 
monotonous and converges to $\Phi_{inf}^+$. 
Similarly, we obtain that if $\eps_0$ is sufficiently small, then there exists $M'>0$ such that for any $\eps<\eps_0$, 
$|\Phi_{inf}^+-\zeta|\leq M'\eps$. 
We define the stopping time 
$$
R^{K,M'}_{ \mathcal{Y}^K_{inf}}=\inf \left\{ t\geq 0, \mathcal{Y}^K_{inf} \not\in [\zeta-(M'+1)\eps,\zeta+(M'+1)\eps] \right\}.
$$

As in the proof of Theorem 3/(c) in~\cite{champagnat2006microscopic}, we can construct a family of Markov jump processes 
$\wt{\mathcal{Y}}^K_{inf}$ with transition rates 
that are positive, bounded, Lipschitz and uniformly bounded away from $0$, for which we can find the following estimate 
(Chapter 5 of Freidlin and Wentzell~\cite{freidlin1984random}): there exists $V'>0$ such that,
$$
 \P(R^{K,M'}_{\mathcal{Y}^K_{inf}}> e^{KV'})=\P(R^{K,M'}_{\wt{\mathcal{Y}}^K_{inf}}>e^{KV'})\underset{K\to+\infty}{\to}1.
$$
We can deal with the process $\mathcal{Y}^K_{sup}$ similarly and find $M''>0$ and $V''>0$ such that
$$
\P(R^{K,M''}_{\mathcal{Y}^K_{sup}}> e^{KV''})\underset{K\to+\infty}{\to} 1,
$$
with 
$$R^{K,M''}_{\mathcal{Y}^K_{sup}}=\inf\Big\{t\geq 0, \mathcal{Y}^K_{sup}(t) \not\in [\zeta-(M''+1)\eps, 
\zeta+(M''+1)\eps]\Big\}.$$
Finally, for $M_1=M' \vee M''$ and $V_1=V' \wedge V''$, we deduce that 
$
\P(R^{K,M_1}_{\mathcal{Y}^K_{inf}} \wedge R^{K,M_1}_{\mathcal{Y}^K_{sup}}>e^{KV_1})\underset{K\to+\infty}{\to}1.
$
Moreover, if 
$R^{K,(M_1+1)\eps}_{A,1} \leq R^{K,\zeta/2}_{A,1} \wedge R^{K,\zeta/2}_{a,2} \wedge T^{K,\eps}_{a,1} \wedge T^{K,\eps}_{A,2},$ then
$$
R^{K,(M_1+1)\eps}_{A,1}\geq R^K_{\mathcal{Y}^K_{inf}} \wedge R^K_{\mathcal{Y}^K_{sup}}. 
$$ Thus
\begin{equation}
 \label{eq_gdedev}
\P(R^{K,\zeta/2}_{A,1} \wedge R^{K,\zeta/2}_{a,2} \wedge T^{K,\eps}_{a,1} \wedge T^{K,\eps}_{A,2} \wedge e^{KV_1} > R^{K,(M_1+1)\eps}_{A,1})\underset{K\to+\infty}{\to}0.
\end{equation}

\noindent
Using symmetrical arguments for the population process $Z^K_{a,2}$, we find $M_2>0$ and $V_2>0$ such that
\begin{equation}
 \label{eq_gdedev2}
\P(R^{K,\zeta/2}_{A,1} \wedge R^{K,\zeta/2}_{a,2} \wedge T^{K,\eps}_{a,1} \wedge T^{K,\eps}_{A,2} \wedge e^{KV_2} > R^{K,(M_2+1)\eps}_{a,2})\underset{K\to+\infty}{\to}0.
\end{equation}
Finally, we set $M=M_1 \vee M_2$ and $V=V_1 \wedge V_2$. Limits \eqref{eq_gdedev} and \eqref{eq_gdedev2} are still true with $M$ and $V$. Thus we have proved that, as long as the size of the $a$-population in Patch $1$ and the size of the $A$-population in Patch $2$ are small and as long as the time is smaller than $e^{KV}$, the processes $Z^{K}_{A,1}$ and $Z^K_{a,2}$ stay close to $\zeta$, i.e. they belong to $[\zeta-(M+1)\eps, \zeta+(M+1\eps)]$.\\

\noindent
Note that if $\eps_0$ is sufficiently small, $R^{K,(M+1)\eps}_{A,1}\leq R^{K,\zeta/2}_{A,1}$ and $R^{K,(M+1)\eps}_{a,2}\leq R^{K,\zeta/2}_{a,2}$ a.s. for all $\eps<\eps_0$. 
So we reduce our study to the time interval 
$$ I_2^{K,\eps}:= \left[0, R^{K,(M+1)\eps}_{A,1} \wedge R^{K,(M+1)\eps}_{a,2} \wedge T^{K,\eps}_{a,1} \wedge T^{K,\eps}_{A,2} \right].$$
\medskip\\
\textbf{Step 2:} In the sequel we study the extinction time of the stochastic processes $(Z^K_{a,1}(t),t\geq0)$ and 
$(Z^K_{A,2}(t),t\geq0)$. We recall that there exists $\eta\in ]0,1/2[$ such that 
 $\eta\eps/2 \leq z_{a,1}^0,z_{A,2}^0 \leq \eps/2$.
Bounding the birth and death rates of $(Z^K_{a,1}(t),t\geq0)$ and $(Z^K_{A,2}(t),t\geq0)$ as previously, 
we deduce that the sum $(Z^K_{a,1}(t)+Z^K_{A,2}(t),t\geq0)$ is stochastically bounded as follows, on the time 
interval $I_2^{K,\eps}$:
$$
\frac{\mathcal{N}^K_{inf}}{K} \preccurlyeq Z^K_{a,1}+Z^K_{A,2} \preccurlyeq \frac{\mathcal{N}^K_{sup}}{K}.
$$
where $\mathcal{N}^K_{inf}$ is a $\N$-valued binary branching process with birth rate
$b+p \frac{\zeta-(M+1)\eps}{\zeta-M\eps}$,
death rate 
$d  +c\zeta +c(M+2)\eps +p $ and initial state $\lfloor \eta \eps K \rfloor $, and
$\mathcal{N}^K_{sup}$ is a $\N$-valued binary branching process with birth rate
$$b\frac{\zeta+\eps(\beta-M-1)}{\zeta-M\eps}+p,$$
death rate 
$$d  +c\zeta -c(M+1)\eps +p \frac{\zeta-(M+1)\eps}{\zeta-M\eps} ,$$ and initial state $\lfloor \eps K \rfloor +1$. \\

\noindent
It remains to estimate the extinction time for a binary branching process $(\mathcal{N}_t,t\geq 0)$ with a birth rate 
$B$ and a death rate $D> B$. Applying \eqref{ext_times} with $i=\lfloor \eta \eps K \rfloor$, we get:
$$
\begin{aligned}
\forall C<(D-B)^{-1}, \quad  &  \P(S_{0}^{\mathcal{N}} \leq C \log(K))\underset{K\to +\infty}{\to} 0,\\
\forall C>(D-B)^{-1}, \quad &  \P(S_{0}^{\mathcal{N}} \leq C \log(K))\underset{K\to +\infty}{\to} 1.
\end{aligned}
$$
Moreover, if 
$$S_{\lfloor \eps K \rfloor}^{\mathcal{N}}:= \inf \{t >0, \mathcal{N}(t)\geq \lfloor \eps K \rfloor\} , $$ 
then
\begin{equation}
\label{eq_extinction}
 \P\left(S_{0}^{\mathcal{N}}>K \wedge S_{\lfloor \eps K \rfloor}^{\mathcal{N}}\right)\underset{K \to +\infty}{\to} 0
\end{equation}
(cf. Theorem 4 in~\cite{champagnat2006microscopic}). Thus
\begin{equation*}
\begin{aligned}
\P&(T^K_0<C \log(K))-\P\left(S_{0}^{\mathcal{N}^K_{inf}}<C \log(K)\right)\\
& \leq \P\left(T^K_0>T^{K,\eps}_{a,1}\wedge T^{K,\eps}_{A,2}\wedge K\right)
+ \P\left(T^{K,\eps}_{a,1}\wedge T^{K,\eps}_{A,2}\wedge K> R^{K,(M+1)\eps}_{A,1} \wedge R^{K,(M+1)\eps}_{a,2}\right)\\
&\leq \P\left(S_{0}^{\mathcal{N}^K_{sup}}>S_{\lfloor \eps K \rfloor}^{\mathcal{N}^K_{sup}}\wedge K\right)
+ \P\left(T^{K,\eps}_{a,1}\wedge T^{K,\eps}_{A,2}\wedge K> R^{K,(M+1)\eps}_{A,1} \wedge R^{K,(M+1)\eps}_{a,2}\right).
\end{aligned}
\end{equation*}

The last term of the last line converges to $0$ when $K$ tends to $0$ according to~\eqref{eq_gdedev} and \eqref{eq_gdedev2}. The first one also tends to $0$ according to~\eqref{eq_extinction}. Thus,
\begin{equation*}
 \lim_{K\to+\infty}\P\left(T^K_0<C \log(K)\right) \leq \lim_{K\to +\infty}\P\left(S_{0}^{\mathcal{N}^K_{inf}}<C \log(K)\right).
\end{equation*}
We prove similarly that
\begin{equation*}
 \lim_{K\to+\infty}\P\left(T^K_0<C \log(K)\right) \geq \lim_{K\to +\infty}\P\left(S_{0}^{\mathcal{N}^K_{sup}}<C \log(K)\right).
\end{equation*}
We conclude the proof by noticing that the growth rates of the processes $\mathcal{N}^K_{inf}$ and $\mathcal{N}^K_{sup}$ 
are equal to $-b(\beta-1)$ up to a constant times $\eps$.
\end{proof}

\subsection{Proof of Theorem \ref{maintheo}}

We can now prove our main result:

Let $\eps$ be a small positive number. 
Applying Lemma \ref{lemapprox} and Theorem \ref{thm_systdet} we get the existence of a positive real number 
$s_\eps$ such that 
$$\lim_{K \to \infty}  \P\left( \|\tb{N}^K(s_\eps) - (\zeta K,0,0,\zeta K)\| \leq \eps K/2 \right)=1 . $$
Using Proposition \ref{propexttime} and the Markov property yield that there exists $C_0>0$ such that 
 $$ \lim_{K \to \infty}\P \left( \left| \frac{T^K_{ \mathcal{B}_{\eps}}}{\log K}-\frac{1}{b(\beta-1)} \right|\leq C_0\eps  \right)= 1,$$
where by definition, we recall that $T^K_{ \mathcal{B}_{\eps}}$ is the hitting time of $ \mathcal{B}_{\eps}$. 
Moreover, the migration rates are equal to zero for any $t\geq T^K_{ \mathcal{B}_{\eps}}$, so $$Z^K_{a,1}(t)=Z^K_{A,2}(t)=0, \ \text{ for any } \ 
t\geq T^K_{ \mathcal{B}_\eps}.$$
After the time $T^K_{ \mathcal{B}_\eps}$, the $A$-population in the patch 1 and the $a$-population in the patch 2 evolve independently from each other according to two logistic birth and death processes with birth rate $b\beta$, death rate $d$ and competition rate $c$. Using Theorem 3(c) in Champagnat~\cite{champagnat2006microscopic}, we deduce that for any $m>1$, there exists $V>0$ such that
$$  \inf_{X \in \mathcal{B}_\eps} \P_{X}(T^K_{\mathcal{B}_{m\eps}} \geq e^{KV}) \underset{K\to+\infty}{\to } 1,$$
which ends the proof.

\section{Influence of the migration parameter $p$: mumerical simulations} \label{sectionillustr}
In this section, we present some simulations of the deterministic dynamical system \eqref{systdet}. 
We are concerned with the influence of the migration rate $p$ on the time to reach a neighbourhood of 
the equilibrium~\eqref{eq-z00z}.
Note that $p$ has no impact on the corresponding relaxation time for the stochastic system, because extinction of the minorities happens on a longer time scale. \\
For any value of $p$, we evaluate the first time $T_{\eps}(p)$ such that the solution $(z_{A,1}(t),z_{a,1}(t),
 \linebreak[4]z_{A,2}(t),z_{a,2}(t))$ to~\eqref{systdet} belongs to the set 
$$
\mathcal{S_\eps}=\left\{(z_{A,1},z_{a,1},z_{A,2},z_{a,2})\in  \R_+^4, (z_{A,1}-\zeta)^2+z_{a,1}^2+z_{A,2}^2+(z_{a,2}-\zeta)^2\leq \eps^2 \right\},
$$
which corresponds to the first time the solution enters an $\varepsilon-$neighbourhood of $(\zeta,0,0,\zeta)$.

In the following simulations, the demographic parameters are given by:
$$\beta=2, \qquad b=2, \qquad d=1 \qquad \text{and} \qquad c=0.1.$$
For these parameters, 
$$\zeta =30 \qquad \text{and} \qquad p_0=\sqrt{5}-1\simeq 1.24.
$$
The migration rate as well as the initial condition vary. \\

\noindent
\textbf{Description of the figures:}
Figure~\ref{fig_tpsrelatif} presents the plots of $p \mapsto T_{\eps}(p)-T_{\eps}(0)$. The simulations are computed with $\eps=0.01$ and with initial conditions  
$(z_{A,1}(0),z_{a,1}(0),z_{A,2}(0),z_{a,2}(0))$ such that $z_{a,1}(0)=z_{A,1}(0)-0.1$ with $z_{A,1}(0)\in\{0.3,0.5,1,2,3,5,10,15\}$ 
and $(z_{A,2}(0),z_{a,2}(0))\in \{(1,30),(15,16)\}$.
Figure~\ref{fig_trajectories} presents the trajectories of some solutions to the dynamical system~\eqref{systdet} in the two 
phase planes which represent the two patches. We use the same parameters as in Figure~\ref{fig_tpsrelatif} and the initial conditions 
are given in the captions.  For each initial condition, we plot the trajectories for three different values of $p$: $ 0,1$ and $20$.\\

\noindent
\textbf{Conjecture}:
First of all, we observe that for all values under consideration, the time $T_{\eps}(p)$ to reach the set 
$\mathcal{S}_{\eps}$ is finite even if $p>p_0$. Therefore, we make the following conjecture:

\begin{conj}
For any initial condition $(z_{1,A}(0),z_{1,a}(0),z_{2,A}(0),z_{2,a}(0))\in \mathcal{D}$, where $\mathcal{D}$ is defined by~\eqref{defdelta},
\begin{equation*}
(z_{1,A}(t),z_{1,a}(t),z_{2,A}(t),z_{2,a}(t)) \underset{t\to +\infty}{\longrightarrow} (\zeta,0,0,\zeta).
\end{equation*}
\end{conj}

\begin{figure} [!h]
\begin{minipage}{0.49\textwidth}
 \begin{center}
  \includegraphics[width=0.99\textwidth]{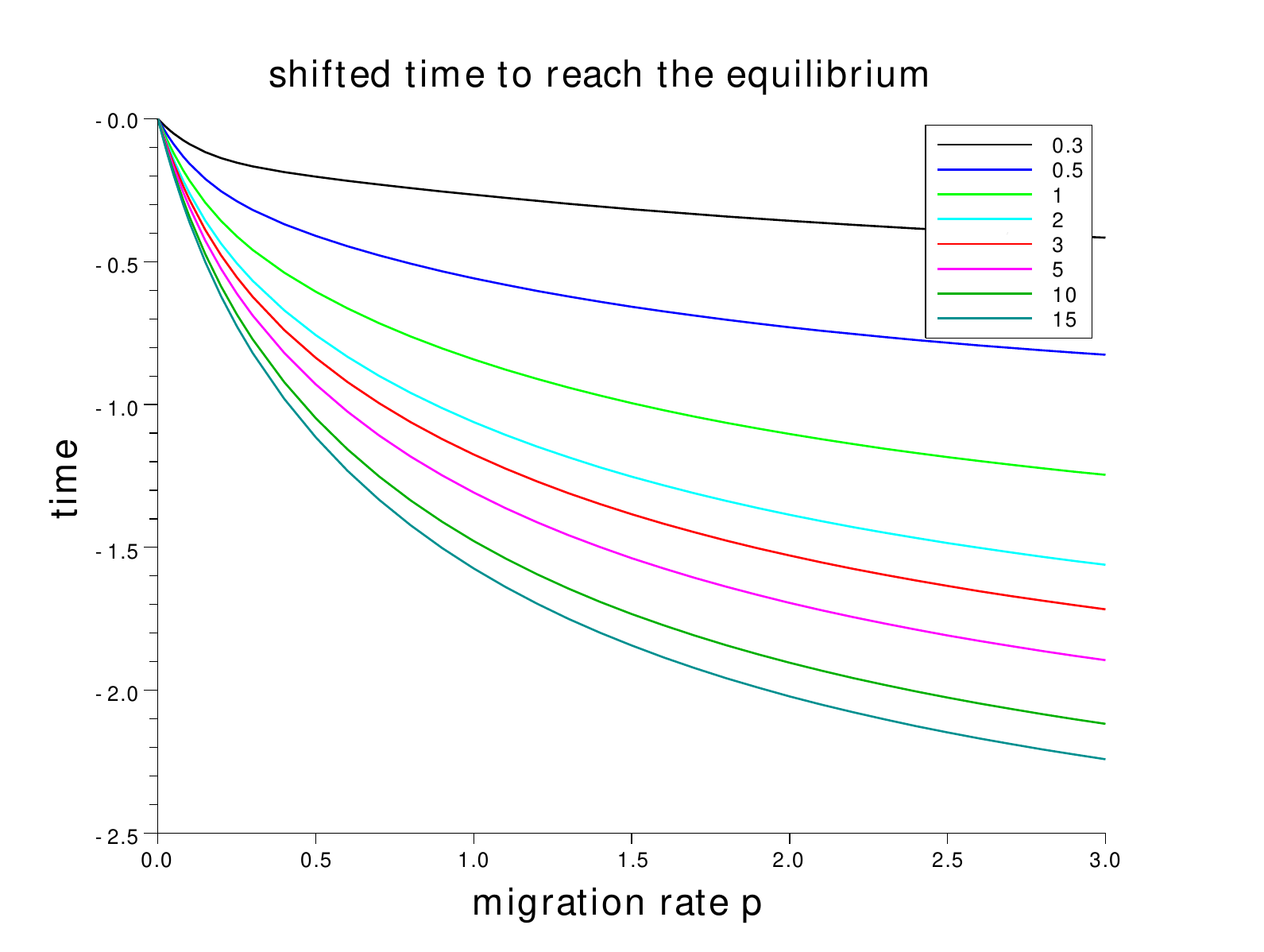}\\
\small{ (a) $(z_{A,2}(0),z_{a,2}(0))=(1,30)$}
 \end{center}
\end{minipage}
\begin{minipage}{0.49\textwidth}
 \begin{center}
\includegraphics[width=0.99\textwidth]{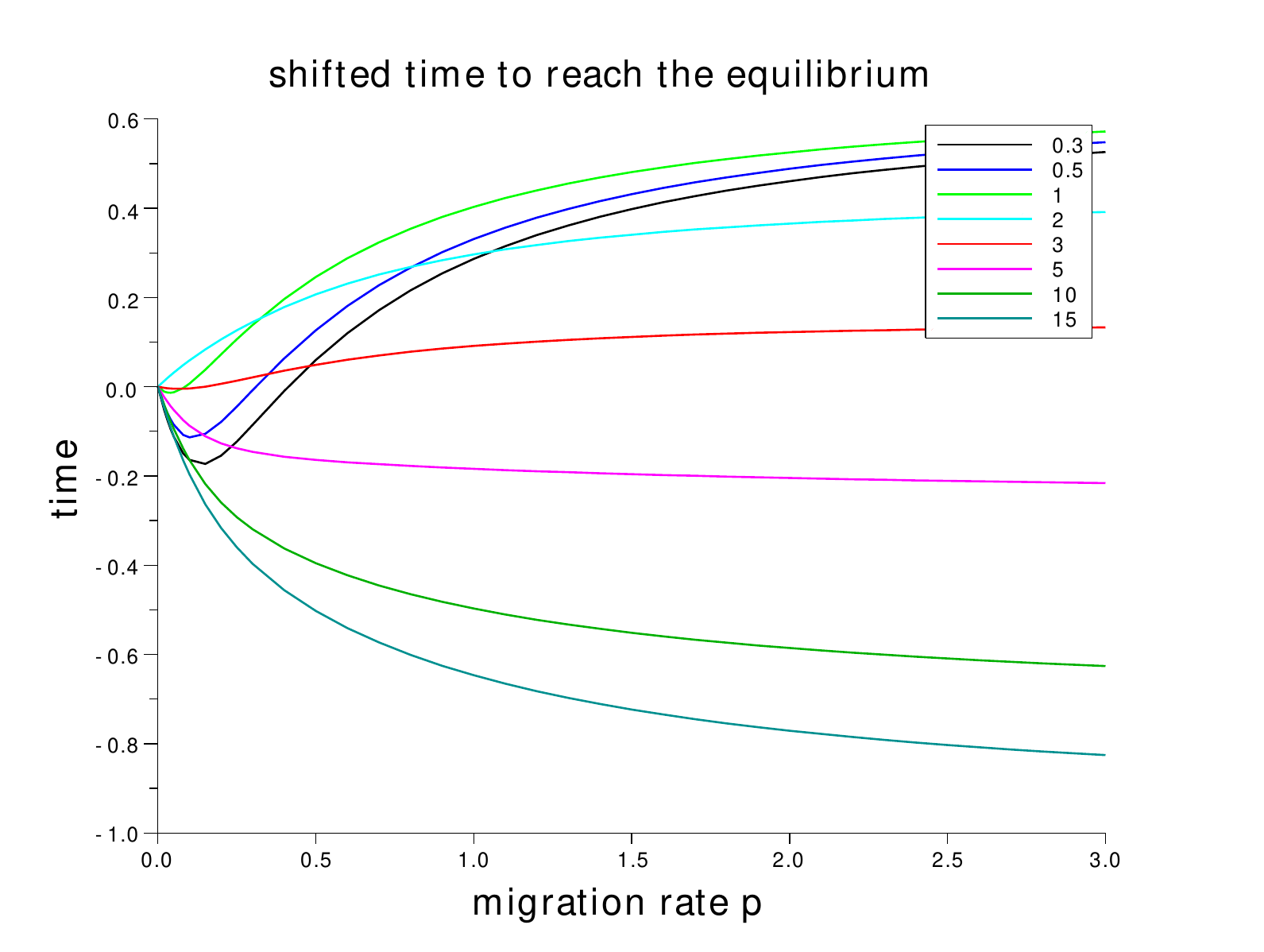}\\
\small{ (b) $(z_{A,2}(0),z_{a,2}(0))=(15,16)$}
 \end{center}
\end{minipage}
\caption{\label{fig_tpsrelatif} \small{{For different values of the initial condition}, we plot
$p \mapsto T_{\eps}(p)-T_{\eps}(0)$. 
The initial condition is $(z_{A,1}(0),z_{A,1}(0)-0.1,z_{A,2}(0),z_{a,2}(0))$ where $z_{A,1}(0)\in\{0.3,0.5,1,2,3,5,10,15\}$ as represented by the colors of the legend;
 and $(z_{A,2}(0),z_{a,2}(0))=(1,30)$ on the left, and $(z_{A,2}(0),z_{a,2}(0))=(15,16)$ on the right. }}
\end{figure}

\noindent 
\textbf{Influence of $p$ when the initial condition in patch $2$ is close to the equilibrium}:
Figure~\ref{fig_tpsrelatif}(a) presents the results for $(z_{A,2}(0),z_{a,2}(0))=(1,30)$, that is if the initial condition in the patch $2$ is close to its equilibrium 
(recall that $\zeta=30$ with the parameters under study). Observe that for any value of $(z_{A,1}(0),z_{a,1}(0)=z_{A,1}(0)-0.1)$, the time for reproductive isolation to occur is reduced when the migration rate is large. Hence, the migration rate seems here to strengthen the homogamy.
This is confirmed by Figure~\ref{fig_trajectories}(a) and (b) where examples of trajectories with the same initial conditions as in Figure~\ref{fig_tpsrelatif}(a) are drawn.
The two Figures~\ref{fig_trajectories}(a) and (b) present similar behaviours: when $p$ increases, the number of $a$-individuals in patch $1$ decreases at any time 
whereas the number and the proportion of $a$-individuals in patch $2$ remain almost constant. 
These behaviours derive from two phenomena.  
On the one hand, the $a$-individuals are able to leave patch $1$ faster when $p$ is large. 
On the other hand, the value of $p$ does not affect the migration outside patch $2$ which is almost zero in view of the 
small proportion of $A$-individuals in the patch $2$.\\

\noindent
\textbf{Influence of $p$ when $a$- and $A$-population sizes are initially similar in patch $2$}:
 On Figure~\ref{fig_tpsrelatif}(b) we 
are interested in the case where the $A$- and 
 $a$- initial populations in patch $2$ have a similar size and the sum $z_{A,2}(0)+z_{a,2}(0)$  is 
close to $\zeta$. 
 Observe that for $z_{A,1}(0)\in \{5,10,15\}$, the time $T_{\eps}(p)$ decreases with respect to $p$ 
but not as fast as previously.
By plotting some trajectories when $z_{A,1}(0)=10$ on Figure~\ref{fig_trajectories}(c), we note that the dynamics is not the same as for the previous case (Fig. ~\ref{fig_trajectories}(a)). Here, a large migration rate affects the migration outside the two patches in such a way that the equilibrium is reached faster.\\
Finally, Figure~\ref{fig_tpsrelatif}(b) also presents behaviours that are essentially different for $z_{A,1}(0)\in\{0.3,0.5,1,2,3\}$. 
In these cases, the migration rate does not strengthen the homogamy. 
We plot some trajectories from this latter case in Figure~\ref{fig_trajectories}(d) where $z_{A,1}(0)=1$.
Observe that a high value of $p$ favors the migration outside patch $2$ for the two types $a$ and $A$ since the proportions of the two alleles in patch $2$ are almost equal at time $t=0$. This is not the case in the patch $1$ where the value of $p$ does not affect significantly the initial migration outside patch $1$ since the population sizes are smaller. Hence, patch $1$ is filled by the individuals that 
flee patch $2$ where the migration rate is high. Therefore, both $a$- and $A$- populations increase at first, but the $A$-individuals remain dominant in patch $1$ and thus the $a$-population is disadvantaged. Finally, the $a$-individuals that flee the patch $2$, 
find a less favorable environment in patch $1$ and therefore the time needed to reach the equilibrium is extended because of the dynamics in patch $1$.\\

\begin{figure}[h!]
 \begin{minipage}{0.49\textwidth}
  
\begin{center}
\begin{minipage}{0.49\textwidth}
  \includegraphics[width=1.18\textwidth]{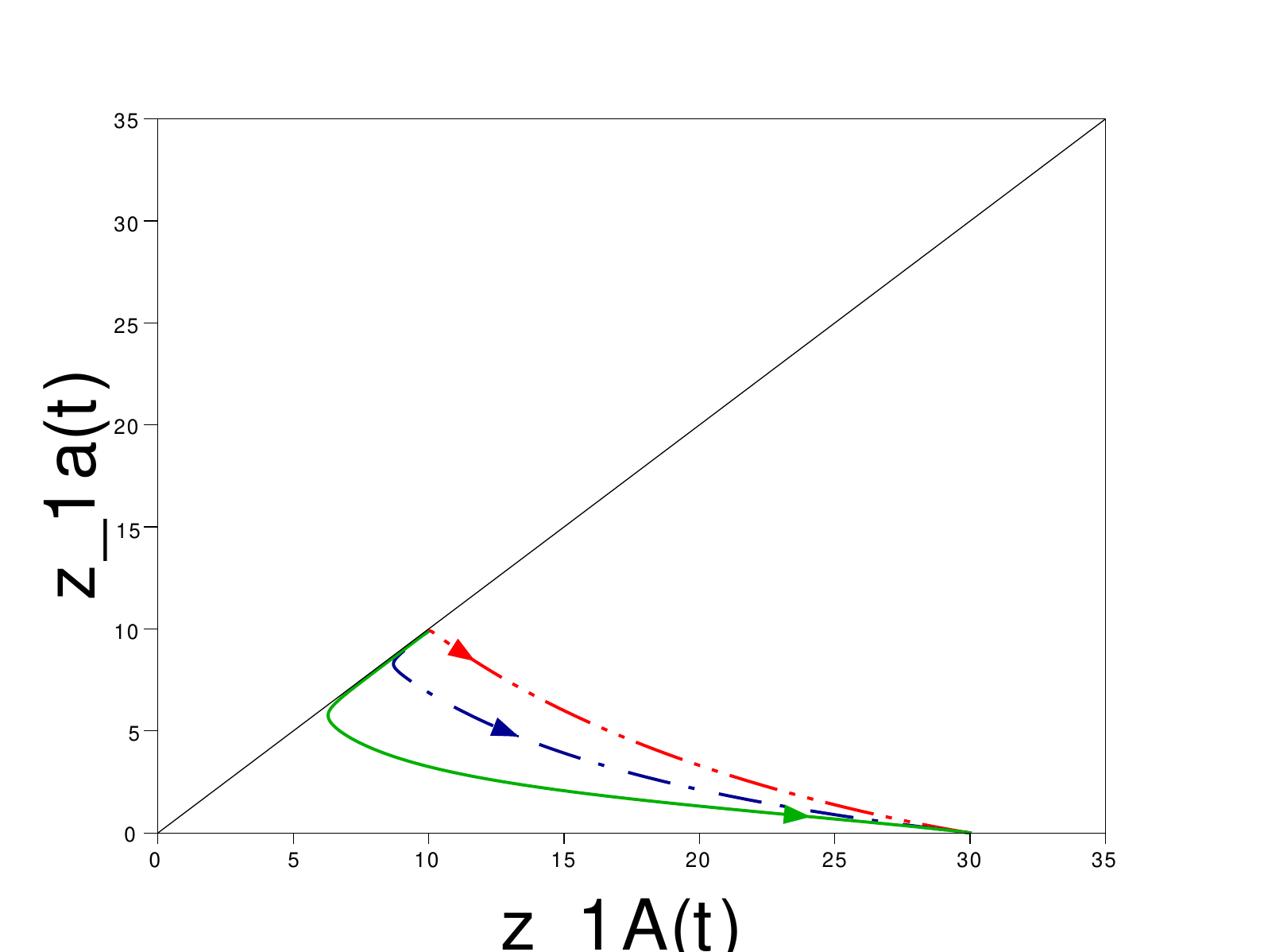}
 \end{minipage}
\begin{minipage}{0.49\textwidth}
  \includegraphics[width=1.18\textwidth]{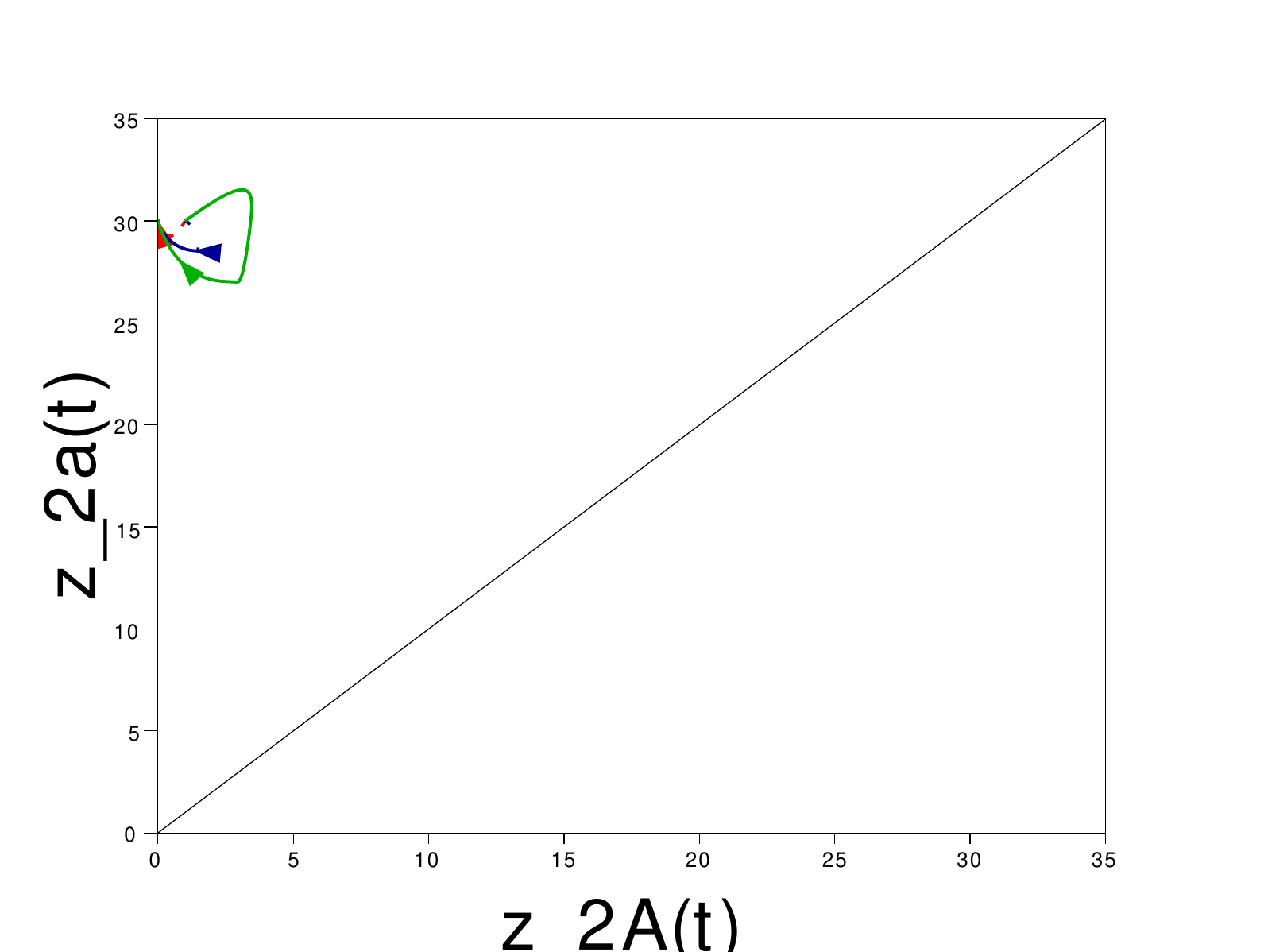}
 \end{minipage}
\small{(a) (10,9.9,1,30)}

  \begin{minipage}{0.49\textwidth}
  \includegraphics[width=1.18\textwidth]{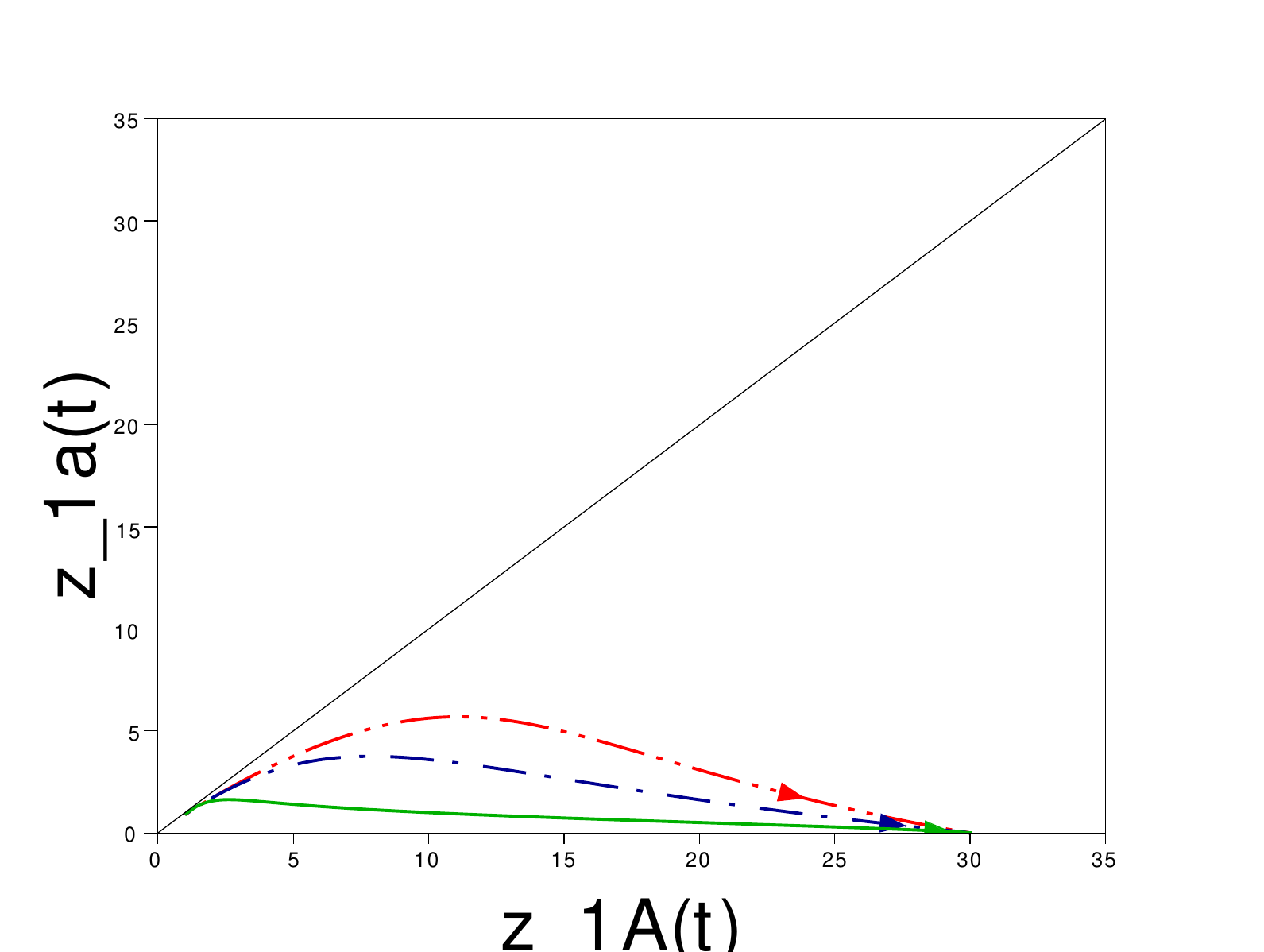}
 \end{minipage}
\begin{minipage}{0.49\textwidth}
  \includegraphics[width=1.18\textwidth]{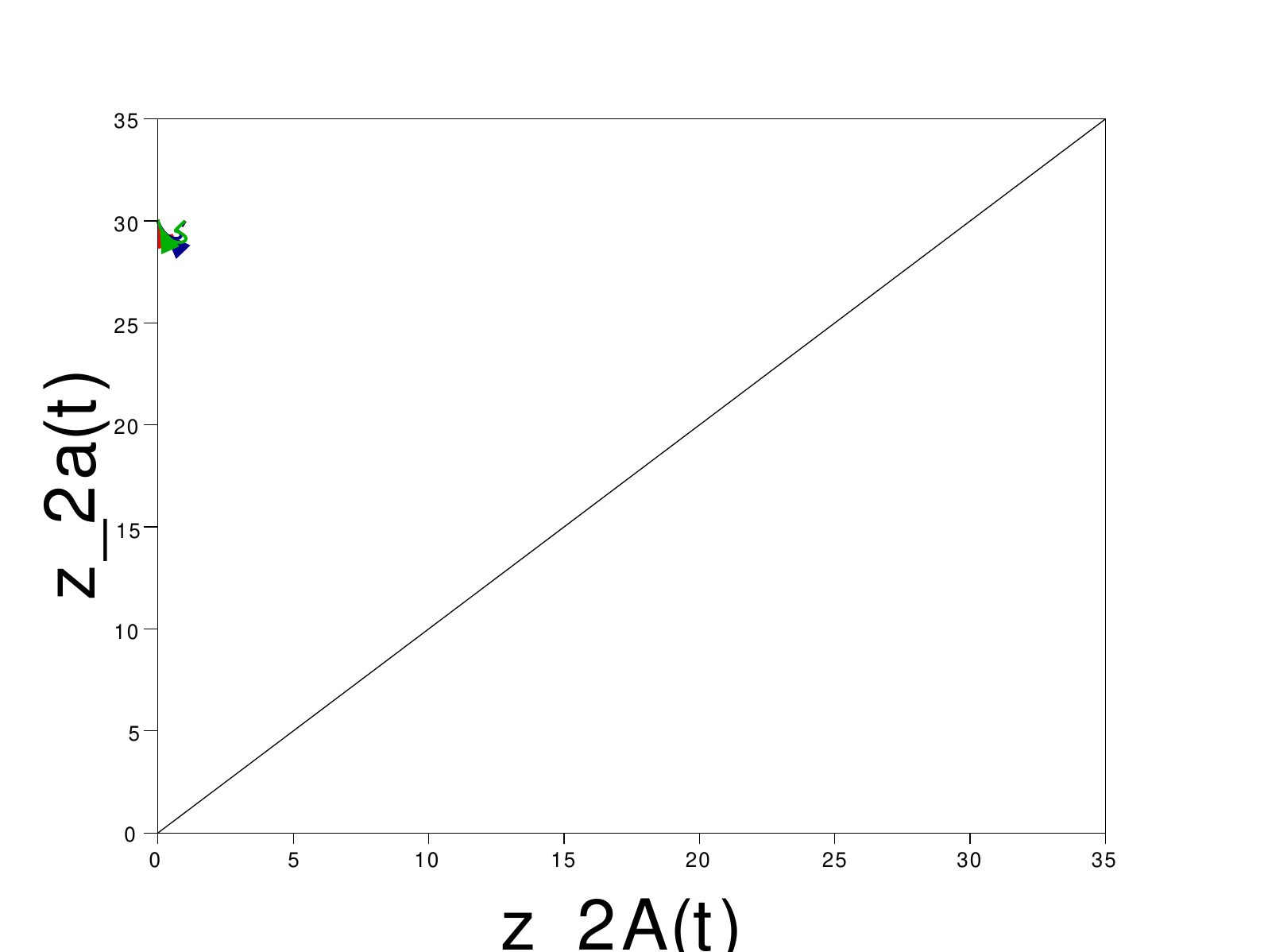}
 \end{minipage}
\small{(b) (1,0.9,1,30)}
\end{center}

 \end{minipage}
\begin{minipage}{0.49\textwidth}
  
\begin{center}
\begin{minipage}{0.49\textwidth}
  \includegraphics[width=1.18\textwidth]{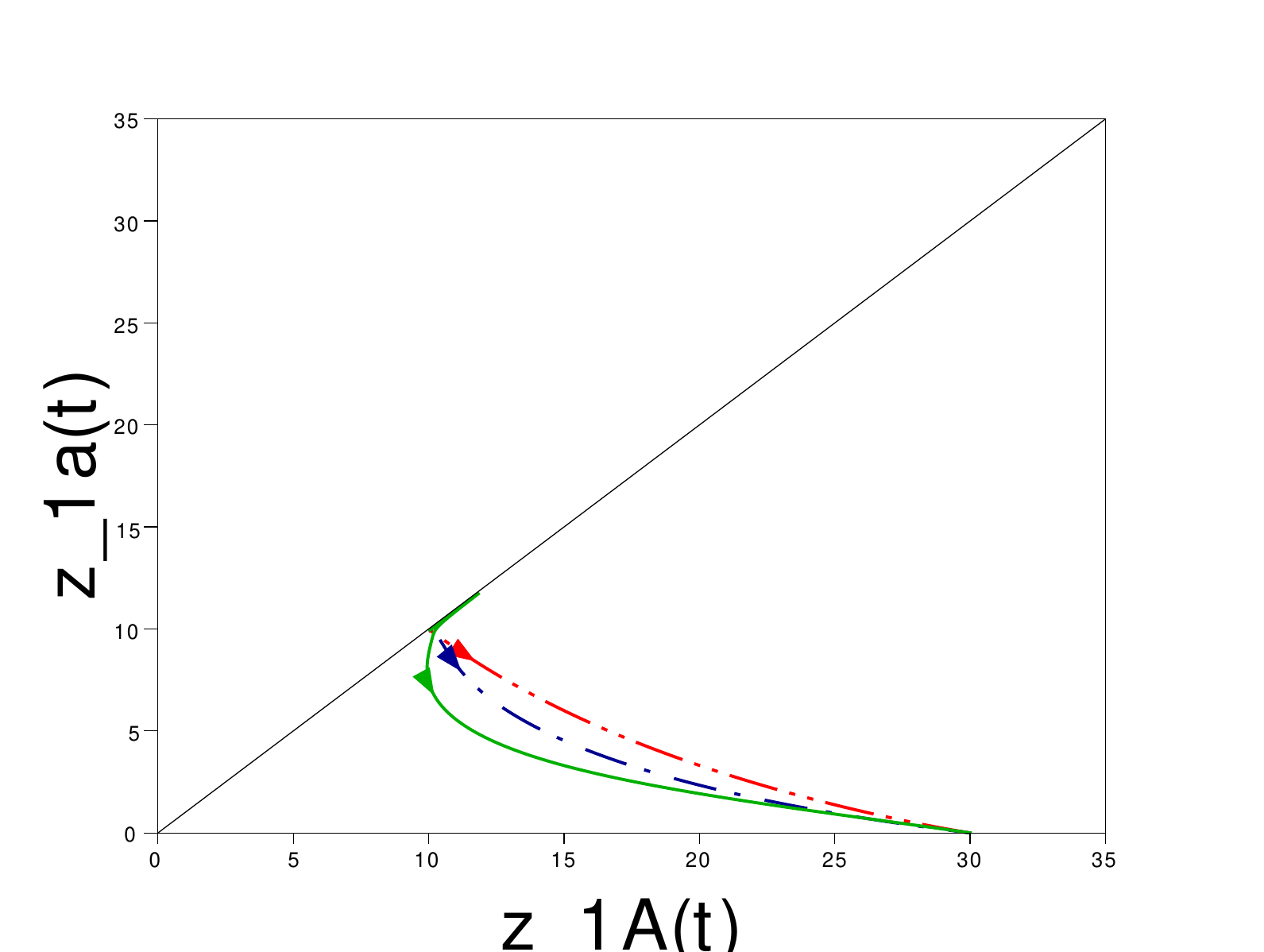}
 \end{minipage}
\begin{minipage}{0.49\textwidth}
  \includegraphics[width=1.18\textwidth]{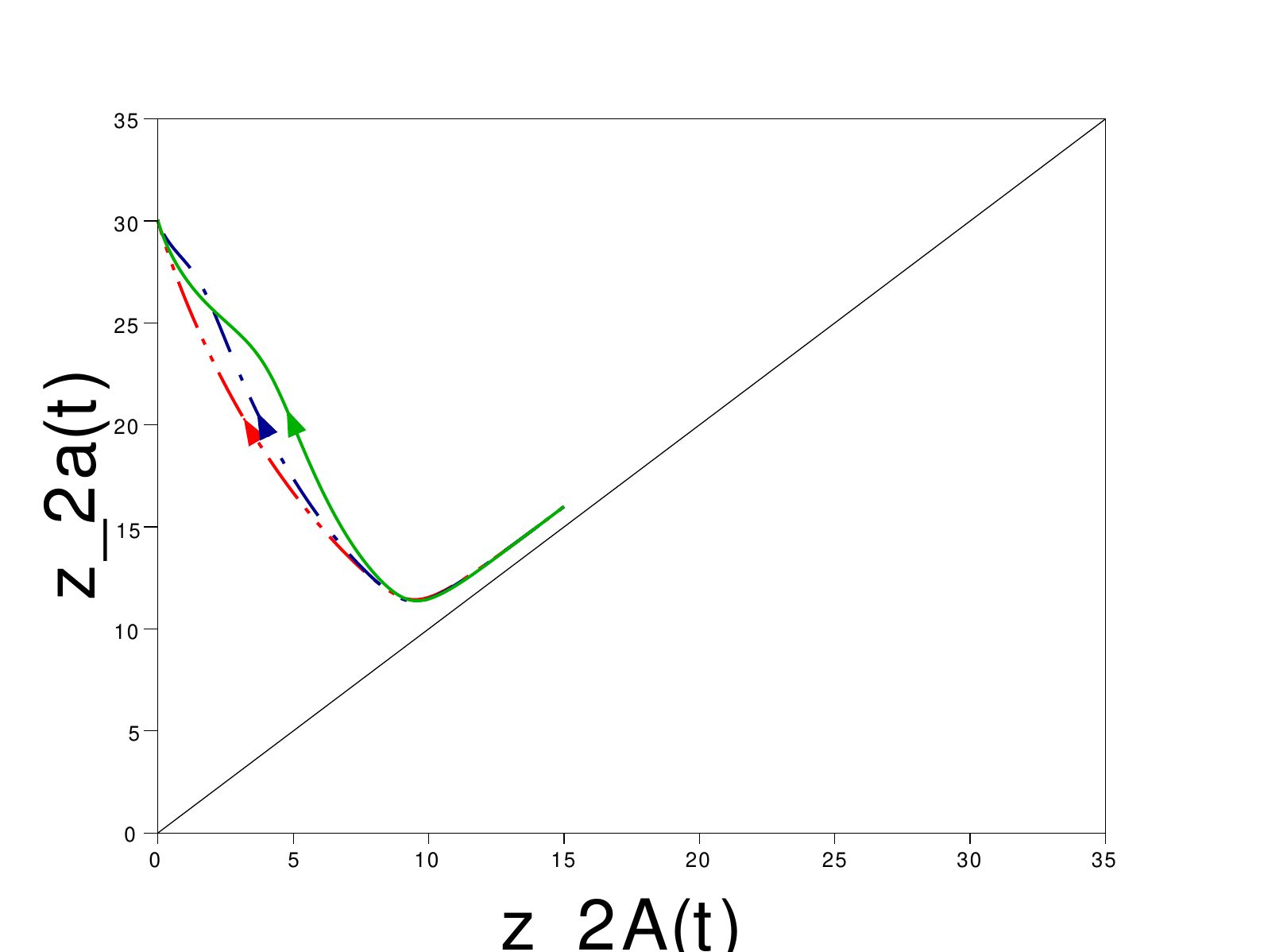}
 \end{minipage}
\small{(c) (10,9.9,15,16)}

\begin{minipage}{0.49\textwidth}
  \includegraphics[width=1.18\textwidth]{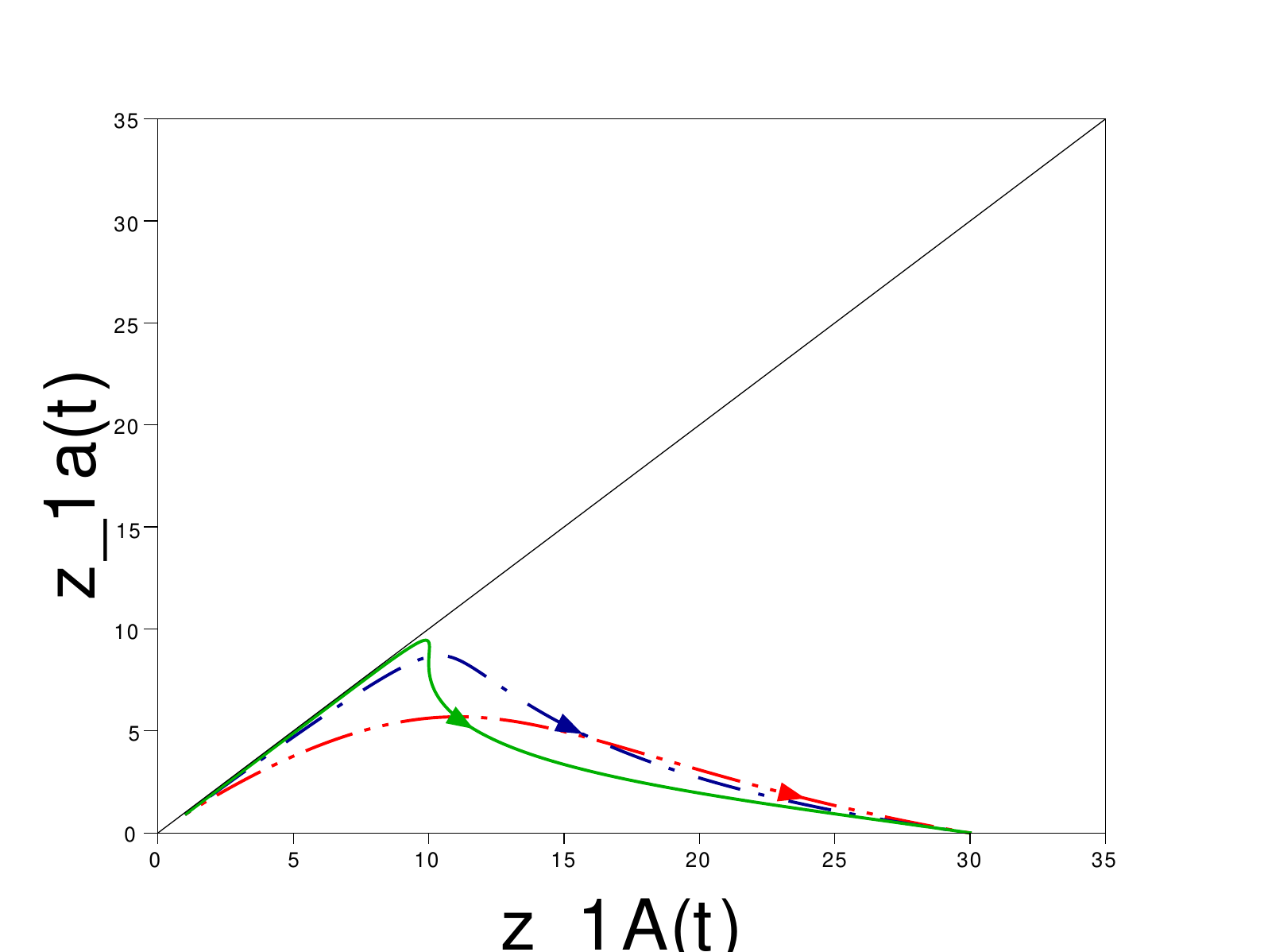}
 \end{minipage}
\begin{minipage}{0.49\textwidth}
  \includegraphics[width=1.18\textwidth]{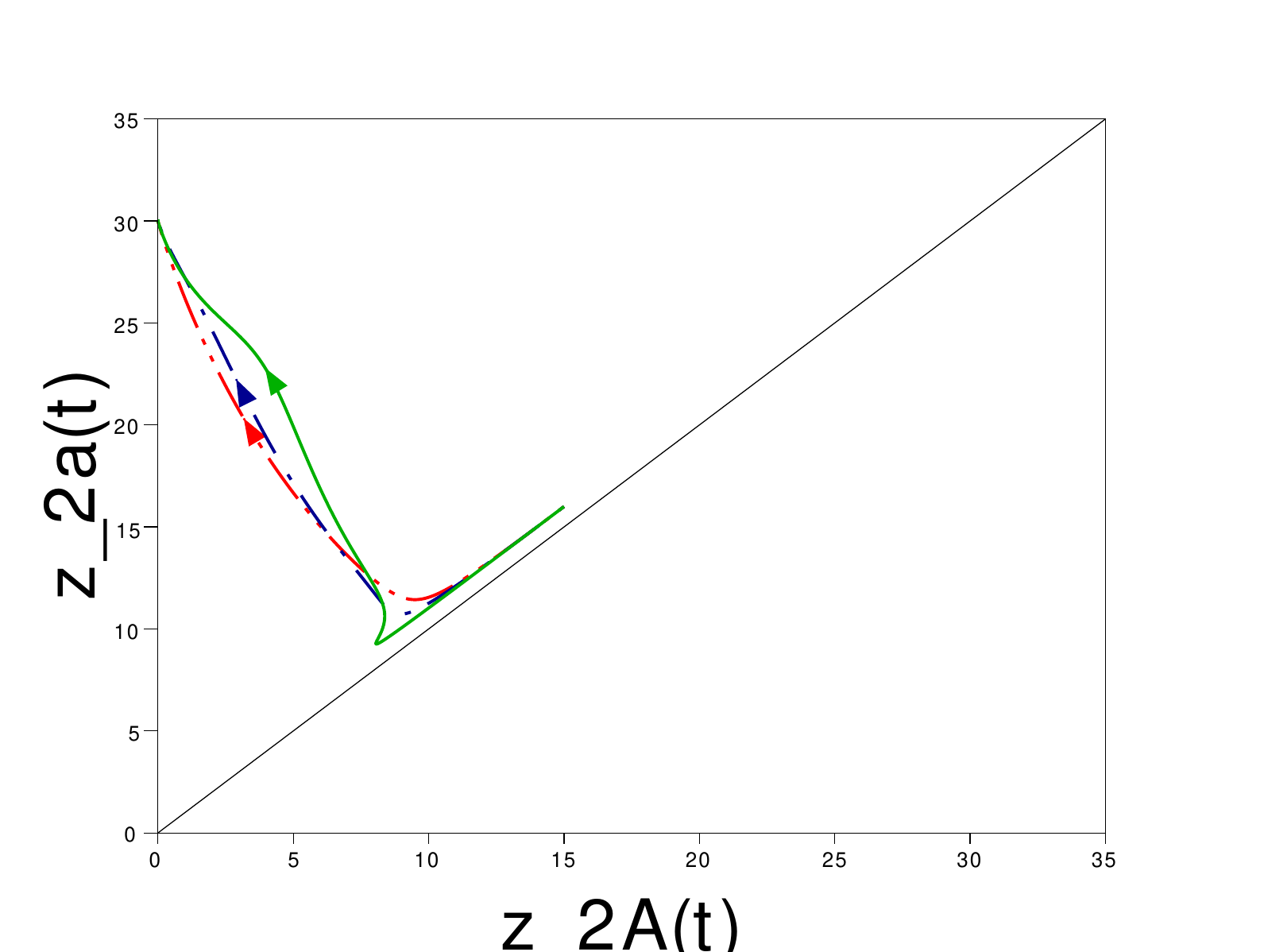}
 \end{minipage}
\small{(d) (1,0.9,15,16)}
\end{center}

 \end{minipage}
\caption{\label{fig_trajectories}\small{For  four different initial conditions, we plot the trajectories in the phase planes  
which represent the patch $1$ (left) and the patch $2$ (right) for $t\in [0,10]$ 
and for three values of $p$: $p=0$ (red), $p=1$ (blue), $p=20$ (green). The initial condition is given under each pair of plots in the format 
$(z_{A,1}(0),z_{a,1}(0),z_{A,2}(0),z_{a,2}(0))$.
Note that the initial conditions on (a) and (c) (resp. (b) and (d)) corresponds to the dark green (resp. light green) curve on Figure \ref{fig_tpsrelatif}(a) and \ref{fig_tpsrelatif}(b).}
}
\end{figure}

\noindent
{\textbf{Conclusion}: }
As a conclusion, similarly to the case of selection-migration model (see e.g. \cite{akerman2014consequences}) migration can have different impacts on the population dynamics. On the one hand, a large migration rate helps the individuals to escape a disadvantageous habitat~\cite{clobert2001dispersal} but there are also risks to move through unfamiliar or less suitable habitat. Thus, a trade-off between the two phenomena explains the influence of $p$ on the time to reach the equilibrium.\\

\section{Generalisations of the model} \label{sectiongeneralisation}
 
Until now we studied a simple model to make clear the important properties allowing to get spatial segregation between patches.
We now prove that our findings are robust by studying some generalisations of the model and showing that we can relax 
several assumptions and still get spatial segregation between patches.

\subsection{Differences between patches}

We assumed that the patches were ecologically equivalent in the sense that the birth, death and competition rates $b$, $d$ and 
$c$, respectively, did not depend on the label of the patch $i \in \mathcal{I}$. In fact we could make these parameters depend on 
the patch, and denote them $b_i$, $d_i$ and 
$c_i$, $i \in \mathcal{I}$. In the same way, the sexual preference $\beta_i$ and the migration rate $p_i$ could depend on 
the label of 
the patch $i \in \mathcal{I}$. 
As a consequence, the dynamical system \eqref{systdet} becomes
\ben\label{systdet2}
\left\{\begin{array}{l} \frac{d}{dt}z_{A,1}(t)=
z_{A,1}\Bigl[ b_1\frac{\beta_1 z_{A,1}+z_{a,1}}{z_{A,1}+z_{a,1}}-d_1-c_1(z_{A,1}+z_{a,1})-p_1\frac{z_{a,1}}{z_{A,1}+z_{a,1}}\Bigr]+p_2\frac{z_{A,2}z_{a,2}}{z_{A,2}+z_{a,2}}\\
\frac{d}{dt}z_{a,1}(t)=
z_{a,1}\Bigl[ b_1\frac{\beta_1 z_{a,1}+z_{A,1}}{z_{A,1}+z_{a,1}}-d_1-c_1(z_{A,1}+z_{a,1})-p_1\frac{z_{A,1}}{z_{A,1}+z_{a,1}}\Bigr]+p_2\frac{z_{A,2}z_{a,2}}{z_{A,2}+z_{a,2}}\\
\frac{d}{dt}z_{A,2}(t)=
z_{A,2}\Bigl[ b_2\frac{\beta_2 z_{A,2}+z_{a,2}}{z_{A,2}+z_{a,2}}-d_2-c_2(z_{A,2}+z_{a,2})-p_2\frac{z_{a,2}}{z_{A,2}+z_{a,2}}\Bigr]+p_1\frac{z_{A,1}z_{a,1}}{z_{A,1}+z_{a,1}}\\
\frac{d}{dt}z_{a,2}(t)=
z_{a,2}\Bigl[ b_2\frac{\beta_2 z_{a,2}+z_{A,2}}{z_{A,2}+z_{a,2}}-d_2-c_2(z_{A,2}+z_{a,2})-p_2\frac{z_{A,2}}{z_{A,2}+z_{a,2}}\Bigr]+p_1\frac{z_{A,1}z_{a,1}}{z_{A,1}+z_{a,1}}.\end{array}\right.
\een
The set $\mathcal{D}$ is still invariant under this new system and the solutions to \eqref{systdet2} with initial conditions in 
$\mathcal{D}$ hit in finite time the invariant set 
\begin{equation*}
 \mathcal{K}'_p:= \left\{ 
\tb{z} \in \mathcal{D}, \; z_{A,i}+z_{a,i}  \in \left[ \frac{b_i(\beta_i+1)-2d_i-p_i}{2c_i}, \zeta_i+\frac{p_{\bar{i}}}{2c_i}, \ 
i \in \mathcal{I}\right] 
\right\},
\end{equation*}
where 
$$ \zeta_i: = \frac{b_i\beta_i- d_i}{c_i}.$$
As $\mathcal{D}$ is an invariant set under \eqref{systdet2}, we can define the function $V$ as in \eqref{defV} for every solution 
of $V$ with initial condition in $\mathcal{D}$. Its first order derivative is 
$$ \frac{d}{dt}V(\tb{z}(t))
= -\underset{i=1,2}{\sum}
   \frac{z_{A,i}z_{a,i}}{z_{A,i}+z_{a,i}}\left[ \frac{2b_i(\beta_i-1)+2p_i}{z_{A,i}+z_{a,i}}- 
   \frac{2p_{i}}{z_{A,\bar i}+z_{a,\bar i}} \right].$$
As a consequence, we can prove similar results to Theorems \ref{theoCvceD} and \ref{maintheo} under the assumption that $p_1$ and $p_2$ satisfy
\begin{equation*}
p_i c_{\bar{i}} (2 c_i \zeta_i + p_{\bar{i}}) < c_i (b_i(\beta_i-1)+p_i) (b_{\bar{i}}(\beta_{\bar{i}}+1)-2d_{\bar{i}}-p_{\bar{i}}), 
\text{ for } i \in \mathcal{I},
\end{equation*}
and where the constant in front of the time $\log K$ is no more $\frac{1}{b(\beta-1)}$ but $\frac{1}{\omega_{1,2}}$ with
\begin{multline*}
\omega_{1,2}=\frac{1}{2}(b_1(\beta_1-1)+p_1+b_2(\beta_2-1)+p_2)\\
-\frac{1}{2}\sqrt{(b_1(\beta_1-1)+p_1-b_2(\beta_2-1)-p_2)^2+4p_1p_2}.
\end{multline*} 
Here, note that the constant does depend on all the parameters. Indeed, since there is no ecological neutrality between the two patches, there do not exist simplifications and balancings as in the previous models.\\

\subsection{Migration}
The migration rates under consideration increase when the genetic diversity increases. Indeed, let us consider
$$H_T^{(i)} := 1-\left[\left(\frac{n_{A,i}}{n_{A,i}+n_{a,i}}\right)^2 + \left(\frac{n_{a,i}}{n_{A,i}+n_{a,i}}\right)^2\right] $$
as a measure of the genetic diversity in the patch $i \in \mathcal{I}$.
Note that $H_T^{(i)} \in [0,1/2]$ is known as the "total gene diversity" in the patch $i$ (see \cite{nei1975molecular} for instance) and is widely 
used as a measure of diversity. When we express the migration rates in terms of this measure, we get
$$ \rho_{\alpha,\bar{i} \to i }(n)= p \frac{n_{A,i}n_{a,i}}{n_{A,i}+n_{a,i}}= \frac{p}{2} (n_{A,i}+n_{a,i})H_T^{(i)} . $$
Hence we can consider that the migration helps the speciation. Let us show that we can get the same kind of result when we consider an arbitrary form for the 
migration rate if this latter is symmetrical and bounded.
We thus consider a more general form for the migration rate. More precisely, 
$$ \rho_{\alpha,\bar{i} \to i }(n)= p(n_{A,\bar{i}}, n_{a,\bar{i}}), $$
and we assume 
$$ p(n_{A,\bar{i}}, n_{a,\bar{i}})=p(n_{a,\bar{i}}, n_{A,\bar{i}})\quad \text{and} \quad p(n_{A,\bar{i}}, n_{a,\bar{i}})
\frac{n_{A,\bar{i}}+n_{a,\bar{i}}}{n_{A,\bar{i}}n_{a,\bar{i}}}<p_0,$$
where $p_0$ has been defined in \eqref{defp0}.
Note that the second condition on the function $p$ imposes that as one of the population sizes goes to $0$, then so does the migration rate. In particular, this condition ensures that the points given by \eqref{eq-z00z} and \eqref{eq-z0z0} are still equilibria of the system.
Theorems \ref{theoCvceD} and \ref{maintheo} still hold with this new definition for the migration rate.

\subsection{Number of patches}
Finally, we restricted our attention to the case of two patches, but we can consider an arbitrary number $N \in \N$ of patches.
We assume that all the patches are ecologically equivalent but that the migrant individuals have 
a probability to migrate to an other patch which depends on the geometry of the system. Moreover, 
we allow the individuals to migrate outside the $N$ patches. 
In other words, 
for $\alpha \in \mathcal{A}$, $i\leq N$, $j\leq N+1$ and $\tb{n} \in (\N^{\mathcal{A}})^{N}$,
\begin{equation*}
 \rho_{\alpha,  i \to j}(\tb{n})= p_{ij}\frac{n_{A,i}n_{a,i}}{n_{A,i}+n_{a,i}},
\end{equation*}
where the "patch" $N+1$ denotes the outside of the system.\\
As a consequence, we obtain the following limiting dynamical system for the rescaled process, when the initial population sizes are 
of order $K$ in all the patches: for every $1 \leq i \leq N$,
\ben\label{systdetNbis}
\begin{aligned} \frac{dz_{A,i}(t)}{dt}&=
z_{A,i}\left[ b\frac{\beta z_{A,i}+z_{a,i}}{z_{A,i}+z_{a,i}}-d-c(z_{A,i}+z_{a,i})-\sum_{j\neq i, j\leq N+1} p_{ij}\frac{z_{a,i}}{z_{A,i}+z_{a,i}}\right]\\
&\qquad
+\sum_{j \neq i, j\leq N} p_{ji}\frac{z_{A,j}z_{a,j}}{z_{A,j}+z_{a,j}}\\
\frac{dz_{a,i}(t)}{dt}&=
z_{a,i}\left[ b\frac{\beta z_{a,i}+z_{A,i}}{z_{A,i}+z_{a,i}}-d-c(z_{A,i}+z_{a,i})-\sum_{j\neq i, j\leq N+1} p_{ij}\frac{z_{A,i}}{z_{A,i}+z_{a,i}}\right]\\
&\qquad
+\sum_{j \neq i, j\leq N} p_{ji}\frac{z_{A,j}z_{a,j}}{z_{A,j}+z_{a,j}}\end{aligned}
\een
For the sake of readability, we introduce the two following notations:
$$ p_{i\to}:=\sum_{j\neq i, j\leq N+1} p_{ij} \quad \text{and} \quad p_{i\leftarrow}:= \sum_{j\neq i, j\leq N} p_{ji}.$$
Let $N_A$ be an integer smaller than $N$  which gives the number of patches with a majority of individuals of type $A$.
We can assume without loss of generality that 

$$ z_{A,i}(0)> z_{a,i}(0),\ \text{for} \ 1 \leq i \leq N_A, \quad \text{and} \quad z_{A,i}(0)<z_{a,i}(0), \ \text{for} \ N_A+1\leq i \leq N .$$
Let us introduce the subset of $(\R_+^\mathcal{A})^{N}$
\begin{equation*}  
\mathcal{D}_{N_A,N}:=\{ \tb{z} \in (\R_+^\mathcal{A})^{N}, z_{A,i}-z_{a,i}>0\ \text{for} \  i \leq N_A, 
\quad \text{and} \quad z_{a,i}-z_{A,i}>0 \ \text{for} \  i >N_A \},
\end{equation*}
We assume that the sequence $(p_{ij})_{i,j\in \{1,..,N\}}$ satisfy :
for all $i\in \{1,..,N\}$,
\begin{equation}
\label{hyp_pij}
p_{i\to}< b(\beta+1)-2d \ \text{ and } \ \frac{b(\beta-1)+p_{i\to}}{2cz +p_{i\leftarrow}}-\sum_{j\neq i, j\leq N+1}\frac{p_{ij}}{b(\beta+1)-2d-p_{j\to}}>0.
\end{equation}

Then we have the following result:

\begin{theo}\label{maintheobis}
 We assume that Assumption~\eqref{hyp_pij} holds. Let us assume that $\tb{Z}^K(0)$ converges in probability to a deterministic vector ${\bf{z}^0}$ belonging to $\mathcal{D}_{N_A,N}$
with $(z_{a,1}^0,z_{A,2}^0)\neq (0,0)$.
Introduce the following bounded set depending on $\eps>0$:
 $$ \mathcal{B}_{N_A,N,\eps}:= \Big([({\zeta}-\eps)K,({\zeta}+\eps)K] \times \{0\} \Big)^{N_A}\times \Big(\{0\} \times [({\zeta}-\eps)K,({\zeta}+\eps)K]\Big)^{N-N_A}. $$
 Then there exist three positive constants $\epsilon_0$, $C_0$ and $m$, and a positive constant 
 $V$ depending on $(m,\eps_0)$ such that if $\eps\leq \eps_0$,
 $$ \lim_{K \to \infty}\P \left( \left| \frac{T^K_{\mathcal{B}_{\eps}}}{\log K}-\frac{1}{b(\beta-1)} \right|\leq C_0\eps, 
 \tb{N}^K\left(T^K_{\mathcal{B}_{N_A,N,\eps}}+t\right)\in \mathcal{B}_{N_A,N,m\eps}\; \forall t \leq e^{VK} \right)= 1,$$
 where $T^K_\mathcal{B}$, $\mathcal{B} \subset \R_+^\mathcal{E}$ is the hitting time of the set $\mathcal{B}$ by the population process $\tb{N}^K$.
\end{theo}

The proof is really similar to the one for the two patches. To handle the deterministic part of the proof, we first show that for every initial condition on $\mathcal{D}_{N_A,N}$, the solution of \eqref{systdetNbis} hits the set 
\begin{equation*}
 \mathcal{K}_{N_A,N}:= \left\{ 
\tb{z} \in \Big((\R_+^*)^{\mathcal{A}}\Big)^{N}, \; \{z_{A,i}+z_{a,i} \} \in \left[ \frac{b(\beta+1)-2d-p_{i\to}}{2c}, 
\zeta+\frac{p_{i\leftarrow}}{2c}\right] \forall i \leq N 
\right\} \cap \mathcal{D}_{N_A,N} .
\end{equation*}
 in finite time, and that this set is invariant under \eqref{systdetNbis}. Then, we conclude with the Lyapunov function
\begin{equation*} 
  \tb{z} \in \mathcal{K}_{N_A,N}  \mapsto \underset{i \leq N_A}{\sum}\ln \left(\frac{z_{A,i}+z_{a,i}}{z_{A,i}-z_{a,i}}\right)
 + \underset{N_A<i \leq N}{\sum}\ln \left(\frac{z_{a,i}+z_{A,i}}{z_{a,i}-z_{A,i}}\right) . \\
 \end{equation*}

\bigskip
As a conclusion, several generalisations are possible and a lot of assumptions can be relaxed in the initial simple model.
We can also combine some of the generalisations for the needs of a particular system.  
However, observe that the mating preference influences the time needed to reach speciation in the same way.

\appendix

\section{Technical results {and reduction of the system}}\label{appendix}

This section is dedicated to some technical results needed in the proofs, {as well as the reduction of the system to the minimal number of effective parameters}.
We first prove the convergence 
when $K$ goes to infinity of the sequence of rescaled processes $\tb{Z}^K$ to the 
solution of the dynamical system \eqref{systdet} stated in Lemma  \ref{lemapprox}.

\begin{proof}[Proof of Lemma \ref{lemapprox}]
The proof relies on a classical result of \cite{EK} (Chapter 11).
Let $\tb{z}$ be in $\N^\mathcal{E}/K$. 
According to \eqref{birthrate}-\eqref{migrationrate}, the rescaled birth, death  and migration rates
\begin{equation}
\label{deftildeb} 
\widetilde{\lambda}_{\alpha ,i}(\tb{z})=\frac{1}{K}\lambda_{\alpha,i}(K\tb{z})=\lambda_{\alpha,i}(\tb{z}), \quad  
\wt{d}_{\alpha, i}(\tb{z})=\frac{1}{K}d^K_{\alpha ,i}(K\tb{z})=\left[ d+cz_{A,i}+cz_{a,i}\right]{z_{\alpha,i}}, \end{equation}
and
\begin{equation*}
\wt{\rho}_{\bar{i}\to i}(\tb{z})=\frac{1}{K}\rho_{\bar{i}\to i}(K\tb{z})=\rho_{\bar{i}\to i}(\tb{z}), \quad (\alpha, i) \in \mathcal{E}
\end{equation*}
are Lipschitz and bounded on every compact subset of $ \N^\mathcal{E}$, and do not depend on the carrying capacity $K$. 
Let $(Y_{\alpha,i}^{(\lambda)},Y_{\alpha,i}^{(d)},Y_{\alpha,i}^{(\rho)},(\alpha,i)\in \mathcal{E})$ be twelve 
independent standard Poisson processes.
From the representation of the stochastic process $(\tb{N}^{K}(t),t\geq0)$ in \eqref{def_poisson} we see that the stochastic process $(\bar{\tb{Z}}^{K}(t), t \geq 0)$ 
defined by
\begin{multline*}
 \bar{\tb{Z}}^{K}(t) =\tb{Z}^K(0)+
 \underset{(\alpha,i)\in \mathcal{E}}{\sum}\frac{\tb{e}_{\alpha,i}}{K}
 \Big[{Y}_{\alpha,i}^{(\lambda)}\Big( \int_0^tK \wt{\lambda}_{\alpha,i}(\bar{\tb{Z}}^{K}({s})) ds\Big)-
{Y}_{\alpha,i}^{(d)}\Big( \int_0^t K \wt{d}_{\alpha,i}(\bar{\tb{Z}}^{K}({s})) ds\Big)\Big]\\
+  \underset{(\alpha,i)\in \mathcal{E}}{\sum} \frac{(\tb{e}_{\alpha,i}-\tb{e}_{\alpha,\bar{i}}) }{K} Y_{\alpha,i}^{(\rho)}
\Big( \int_0^t K \wt{\rho}_{\alpha,i}(\bar{\tb{Z}}^{K}({s})) ds\Big),
\end{multline*}
has the same law as $(\tb{Z}^{K}(t), t \geq 0)$. 
Moreover, a direct application of Theorem 2.1 p 456 in \cite{EK} gives that $(\bar{\tb{Z}}^{K}(t), t \leq T)$
converges in probability to $(\tb{z}^{(\tb{z}^0)}(t), t \leq T)$ for the uniform norm.
As a consequence, $(\tb{Z}^K(t), t \leq T)$ converges in law to $(\tb{z}^{(\tb{z}^0)}(t), t \leq T)$ for the same norm.
But the convergence in law to a constant is equivalent to the convergence in probability to the same constant. The result follows.
\end{proof}

We now recall a well known fact on branching processes which 
can be found in \cite{athreya1972branching} p 109.

\begin{lem}
\begin{itemize}
  \item Let $Z=(Z_t)_{t \geq 0}$ be a birth and death process with individual birth and death rates $b$ and $d $. For $i \in \Z^+$, 
$T_i=\inf\{ t\geq 0, Z_t=i \}$ and $\P_i$ is the law of $Z$ when $Z_0=i$. 
If $d\neq b \in \R_+^*$, for every $i\in \Z_+$ and $t \geq 0$,
\begin{equation} \label{ext_times} \P_{i}(T_0\leq t )= \Big( \frac{d(1-e^{(d-b)t})}{b-de^{(d-b)t}} \Big)^{i}.\end{equation}
\end{itemize}
\end{lem}

{As we mentioned in Section \ref{sectionmodel}, it is possible to reduce the number of parameters $b$, $c$, $d$, $p$, $\beta$
by using a change of time and a scaling. Let us introduce the new variables  
$$ \tilde{z}_{\alpha,i}(t):= \frac{c}{b} z_{\alpha,i} \Big(\frac{t}{b}\Big), $$ 
for all $\alpha \in \{A,a\}$, $i \in \{1,2\}$ and $t \geq 0$, and the parameters 
$$\tilde{d}:= \frac{d}{b}, \quad \tilde{p}:= \frac{p}{b}. $$
Then the new variables satisfy the following dynamical system
\begin{equation*}
\frac{d}{dt}\tilde{z}_{\alpha,i}(t)=
\tilde{z}_{\alpha,i}\Bigl[ \frac{\beta \tilde{z}_{\alpha,i}+\tilde{z}_{\bar\alpha,i}}{\tilde{z}_{\alpha,i}+\tilde{z}_{\bar\alpha,i}}
-\tilde{d}-(\tilde{z}_{\alpha,i}+\tilde{z}_{\bar\alpha,i})-
\tilde{p}\frac{\tilde{z}_{\bar\alpha,i}}{\tilde{z}_{\alpha,i}+\tilde{z}_{\bar\alpha,i}}\Bigr]+\tilde{p}\frac{\tilde{z}_{\alpha,\bar i}\tilde{z}_{\bar\alpha,\bar i}}{\tilde{z}_{\alpha,\bar i} +\tilde{z}_{\bar\alpha,\bar i}},
\end{equation*}
for $\alpha\in \{A,a\}$, $\bar\alpha =\{A,a\}\setminus \alpha$, $i \in \{1,2\}$ and $\bar i=\{1,2\} \setminus i$.}

\vspace{.5cm}

{\bf Acknowledgements:} {\sl  The authors would like to warmly thank Sylvie M\'el\'eard for her continual guidance during their respective thesis works. 
They would also like to thank Pierre Collet for his help on the theory of dynamical systems, Sylvain Billiard for many 
fruitful discussions on the biological relevance of their model, and Violaine Llaurens for her help during the revision of the manuscript.
C. C. and C. S. are grateful to the organizers of 
"The Helsinki Summer School on Mathematical Ecology and Evolution 2012: theory of speciation"
which motivated this work.
This work  was partially funded by the Chair "Mod\'elisation Math\'ematique et Biodiversit\'e" of VEOLIA-Ecole Polytechnique-MNHN-F.X, and was also supported by a public grant as part of the
Investissement d'avenir project, reference ANR-11-LABX-0056-LMH,
LabEx LMH}

\bibliographystyle{abbrv}
\bibliography{biblio_speciation}
\end{document}